\documentclass[prl,aps,
10pt,
twocolumn,
preprintnumbers,
nofootinbib,
floatfix,
superscriptaddress]{revtex4-2}
\usepackage{float}
\usepackage{graphicx}
\usepackage{soul}
\usepackage{tikz}
\usetikzlibrary{decorations.pathreplacing}
\makeatletter
\newcommand\myitem[1][]{\item[#1]\def\@currentlabel{#1}}
\makeatother

\usepackage{titlesec}
\usepackage{tocloft}

\newcommand{\myparagraph}[1]{%
  \section{#1}
}

\usepackage{amsmath,amssymb}
\usepackage{mathtools}
\DeclarePairedDelimiter{\floor}{\lfloor}{\rfloor}
\usepackage{xcolor}
\definecolor{darkpurple}{rgb}{0.5, 0.2, 0.8}
\definecolor{darkblue}{rgb}{0.0, 0.0, 0.8}
\definecolor{darkgreen}{rgb}{0.0, 0.4, 0.0}
\definecolor{darkred}{rgb}{0.5, 0.0, 0.0}

\usepackage{amsthm}
\newtheorem{lemma}{Lemma}
\newtheorem{theorem}{Theorem}

\usepackage{hyperref}
\hypersetup{%
    linktocpage,
    colorlinks,
    citecolor=darkgreen,
    linkcolor=darkgreen,
    urlcolor=darkgreen,
    filecolor=darkgreen,
    pdftitle={The Landau Bootstrap},
	pdfdisplaydoctitle=true,
}

\usepackage[normalem]{ulem}

\newcommand{\Li}{{\mathrm{Li}}}
\def\S{\mathcal{S}}
\def\I{\mathcal{I}}
\def\D{\mathrm{D}}
\def\L{\mathrm{L}}
\def\E{\mathrm{E}}
\def\d{d}

\begin{document}
\setlength{\parskip}{0pt}

\title{The Landau Bootstrap}

\author{Holmfridur~S.~Hannesdottir}%
\affiliation{Institute for Advanced Study, Einstein Drive, Princeton, NJ 08540, USA}
\author{Andrew~J.~McLeod}%
\affiliation{Higgs Centre for Theoretical Physics, School of Physics and Astronomy, \\ The University of Edinburgh, Edinburgh EH9 3FD, Scotland, UK}
\author{Matthew~D.~Schwartz}%
\affiliation{Department of Physics, Harvard University, Cambridge, MA 02138, USA}
\author{Cristian~Vergu}%
\affiliation{Institute for Gravitation and the Cosmos, Department of Physics, \\ Pennsylvania State University, University Park, Pennsylvania 16802, USA}

\begin{abstract}
We advocate a strategy of bootstrapping Feynman integrals from just knowledge of their singular behavior. This approach is complementary
to other bootstrap programs, which exploit non-perturbative constraints such as unitarity, or amplitude-level constraints such as gauge invariance. We begin by studying where a Feynman integral can become singular, and the behavior it exhibits near these singularities. We then characterize the space of functions that we expect the integral to evaluate to, in order to formulate an appropriate ansatz. Finally, we derive constraints on where each singularity can appear in this ansatz, and use information about the expansion of the integral around singular points in order to determine the value of all remaining free coefficients. Throughout, we highlight how constraints that have previously only been derived for integrals with generic masses can be extended to integrals involving particles of equal or vanishing mass. We illustrate the effectiveness of this approach by bootstrapping a number of examples, including the four-point double box with a massive internal loop. 

\end{abstract}

\maketitle
The computation of scattering amplitudes in perturbation theory has turned out to be an extraordinarily challenging task. Despite the need for increasingly-precise theory predictions at the LHC~\cite{Huss_2023}, our ability to compute the scattering amplitudes needed to make these predictions remains severely limited, especially for processes that involve multiple heavy species of particles. Evaluating these amplitudes using traditional methods may ultimately prove to be an insurmountable challenge. However, there remains hope that a more indirect approach may still be possible, in which the functional form of amplitudes are bootstrapped directly from their expected mathematical properties. The feasibility of such an approach has greatly increased in recent years, during which rapid progress has been made understanding the types of mathematical structure exhibited by perturbative amplitudes---from the classes of numbers and special functions they evaluate to~\cite{Broadhurst:1998rz,Bogner:2007mn,Bourjaily:2022bwx}, to the ways in which their analytic structure is constrained by basic physical principles~\cite{Landau:1959fi,ELOP,Steinmann,Steinmann2,pham} and the surprising discovery of new number-theoretic symmetries~\cite{Arkani-Hamed:2012zlh,Schlotterer:2012ny,Brown:2015fyf,Panzer:2016snt,Schnetz:2017bko,Caron-Huot:2019bsq,Gurdogan:2020ppd,Dixon:2021tdw,Dixon:2022xqh}. Indeed, a number of nontrivial amplitudes and Feynman integrals have already been determined just from knowledge of these types of mathematical properties (see for instance~\cite{Caron-Huot:2018dsv,Caron-Huot:2020bkp,Morales:2022csr,Dixon:2022rse}), thereby sidestepping the difficult integration problems that have historically plagued amplitude computations. 

In this work, we propose a new bootstrap strategy for computing individual Feynman integrals---the universal building blocks that enter perturbative computations---from just knowledge of the location and nature of their singularities.\footnote{For previous work on bootstrapping individual Feynman integrals, see~\cite{Chicherin:2017dob,Caron-Huot:2018dsv,Henn:2018cdp,He:2021fwf,He:2021eec}.} We refer to this new approach as the {\bf{Landau bootstrap}}. 
In particular, the strategy we propose can be conceptually broken down into seven steps:
\begin{enumerate}
\itemsep.5em 
    \myitem[(i)]\label{step_sing} Identify the location of potential \textbf{singularities} in the space of external kinematics.
    \myitem[(ii)]\label{step_local} Determine the \textbf{local nature} of the integral around these singular points---for instance, discerning whether it develops a pole or a branch cut at a given singular point, by studying the expansion of the integral around this point.
    \item[(iii)] Probe the types of special functions and algebraic prefactors that can appear in the integral, through the computation of \textbf{leading singularities}.
    \item[(iv)] Parametrize the \textbf{space of functions} that the Feynman integral is expected to evaluate to.
    \myitem[(v)]\label{step_global} Determine the \textbf{global nature} of the singularities that appear in the integral, for example determining whether they give rise to branch cuts on the physical sheet.
    \myitem[(vi)]\label{step_rel} Derive constraints on the \textbf{sequential discontinuities} that can appear in the integral.
    \item[(vii)] In the space of functions that remains, find the \textbf{unique function} that is consistent with all expected properties of the Feynman integral.
\end{enumerate}
As we will see, this approach will largely be facilitated by
the types of constraints derived in step~\ref{step_rel}, which go back over sixty years~\cite{Steinmann,Steinmann2,pham,landshoff1966hierarchical} but whose power is only now beginning to be appreciated. In what follows, we discuss each of these steps in more detail, while summarizing known results that hold for integrals involving generic masses and illustrating how these results can be extended to integrals that involve massless particles and particles of equal mass.

\myparagraph{\bf Singularities of Feynman Integrals}
To every Feynman diagram $G$, we can associate an integral of the form
\begin{equation}
    \I (p_i,m_e) = \int \prod_{j=1}^\L \d^\D k_j \frac{N(p_i,k_j)}{\prod_{e=1}^\E (q_e^2-m_e^2+i\varepsilon)} \,,
    \label{eq:feynint_def}
\end{equation}
where $\L$ and $\E$ denote the number of loops and edges in $G$, and $m_e$ and $q_e$ are the mass and momentum flowing through edge $e$ of the graph. The momenta $q_e$ depend linearly on the external momenta $p_i$ and the loop momenta $k_j$, the latter of which are integrated over. In general, the numerator $N(p_i,k_j)$ will be a polynomial in the external and loop momenta. We can either include the numerator directly in our analysis (see Appendix~\ref{appendix:asymptotics_with_numerators}), or reduce these integrals to linear combinations of scalar Feynman integrals (see for example~\cite{Weinzierl:2022eaz}). In this paper, we focus on examples of scalar Feynman integrals for which $N(p_i,k_j)=1$.

The values of $p_i$, $m_e$, and $k_j$ for which the integral in equation~\eqref{eq:feynint_def} can become singular are characterized by the Landau equations~\cite{nakanishi1959,Landau:1959fi,Bjorken:1959fd}
\begin{subequations}
\begin{align}
    \alpha_e (q_e^2-m_e^2) = 0 \,,
    \label{eq:Landau1}
    \\
    \sum_{\text{loop } i} \pm \alpha_e q_e^\mu = 0\,,
    \label{eq:Landau2}
\end{align}
\label{eq:Landau}
\end{subequations}
\unskip
where the $\pm$ sign in front of each term in equation~\eqref{eq:Landau2} is determined by whether the edge $q_e^\mu$ is oriented in the same or opposite direction to the loop momentum $k_j^\mu$ (for more details, see for instance~\cite{ELOP,Sterman:1993hfp}).
Solutions to these equations that exist for all values of $p_i$ and $m_e$ describe infrared singularities, while solutions that only exist for restricted values of these variables describe the singularities that can be encountered as the integral is analytically continued in the space of external kinematics. 

In the non-perturbative S-matrix bootstrap program, solutions to the Landau equations are taken to describe surfaces in the space of external Mandelstam invariants, while $m_e^2$ and $p_i^2$ correspond to physical masses which cannot be complexified. However, the derivation of the Landau equations makes no distinction between external momenta and internal masses; Feynman integrals can become singular at special values of any of the variables on which the integral depends. In fact, singularities at $m_e^2=0$ appear ubiquitously in Feynman integrals, and correspond to Landau diagrams in which all edges but one are contracted to a point.
Solutions also exist for infinite values of the loop momenta; these are usually referred to as \emph{second-type singularities}. More generally, resolving all of the singularities of Feynman integrals generally requires blowing up the integrand before deriving the Landau equations. 
For recent work on this topic, see~\cite{Fevola:2023kaw,Fevola:2023fzn,Helmer:2024wax,Caron-Huot:2024brh}.

In the simplest cases, solutions to the Landau equations describe a \emph{simple pinch}. This means that the Landau equations uniquely determine the value of the internal momenta $q_e$ and the parameters $\alpha_e$ that describe how the integration contour is being pinched (up to simultaneous rescaling of all the $\alpha_e$), and additionally that the Hessian matrix associated with this singularity is negative definite (see Appendix~E of~\cite{Hannesdottir:2022xki}).  Simple pinch singularities enjoy a kind of topological stability which makes them especially amenable to analysis with homological methods. On the other hand, non-simple pinches---which also appear ubiquitously in Feynman integrals---require more care to analyze (for more details, see~\cite{pham,Berghoff:2022mqu}). 
Feynman integral can also have apparent singularities that appear for all values of external momenta, called \emph{permanent pinches}~\cite{boyling1968homological}. The most well-known examples are infrared divergences, which must be removed by regularization (see~\cite{PhysRevD.13.1573} for a discussion of infrared divergences from the point of view of singularities). Other permanent pinches arise in finite integrals from a poor choice of coordinates and can be resolved using blowups.
Resolving these permanent pinches can reveal new singularities that are not identified by the Landau equations as seen in equations~\eqref{eq:Landau1} and~\eqref{eq:Landau2}.

\myparagraph{\bf Local Nature of Singularities}
Once the location of a singularity has been identified, it is often possible to the determine the behavior of a Feynman integral in its vicinity. 
For example, for integrals involving generic masses, Landau showed that the leading non-analytic term in the expansion around first-type singularities can be predicted just from knowledge of which propagators must be put on shell to access this singularity~\cite{Landau:1959fi}. 
More specifically, near such a singular kinematic surface $\varphi=0$, these integrals behave as
\begin{equation} \label{eq:Landau_expansion}
    \I(\varphi) \sim \begin{cases}
   C  \varphi^\gamma \log \varphi   & \text{if } \gamma \in \mathbb{Z}^+ \\ 
   C \varphi^\gamma & \text{otherwise,}
    \end{cases}
\end{equation}
where $C$ is a function of kinematics that is regular in the $\varphi \to 0$ limit, and the \emph{Landau exponent} $\gamma = \frac12(\ell\, \D {-} n{-}1)$ is determined by the Landau diagram that describes which propagators are involved in pinching the contour. More specifically, $\ell$ and $n$ are the number of loops and edges in the graph corresponding to this Feynman integral after all the edges that do not participate in the pinch have been contracted out. 
Note that the branch points that appear in equation~\eqref{eq:Landau_expansion} are only of square root or logarithmic type. This is consistent with all known examples of Feynman integrals (see for example~\cite{Bourjaily:2022vti}), although cube root branch points have been shown, by explicit computations, to appear in energy correlators~\cite{Chicherin:2024ifn}.

Even in cases for which~\eqref{eq:Landau_expansion} does not apply, the leading non-analytic behavior of Feynman integrals near singular configurations will take the form
\begin{equation}
    \I (\varphi) \sim C \varphi^a \log^b \varphi \,,
    \label{eq:genbehavior}
\end{equation}
where $C$ is again regular in the $\varphi \to 0$ limit, and the numbers $a$ and $b$ can be determined by directly expanding the Feynman integral. For example, if the singularity corresponds to a simple pinch, we can follow the derivation in Landau's original paper to determine $a$ and $b$~\cite{Landau:1959fi} (see also~\cite{pham2011singularities}); in other cases, it may also be possible to determine $a$ and $b$ using the method of regions~\cite{Jantzen:2012mw,Heinrich:2021dbf}. As an example, in four dimensions the Regge limit of all four-point Feynman integrals have been characterized; these integrals behave as $s^a\log^bs$ for integer $a$ and $b$ as $s \to \infty$ for fixed $t$ and internal masses~\cite{PhysRev.131.480}. This means that algebraic singularities are not expected to appear in the $s \to \infty$ limit. This correspondingly puts constraints on what types of branch cuts can appear in these Feynman integrals in general kinematics.

\myparagraph{\bf Space of Functions}
Once the location and nature of the singularities of a Feynman integral are known, we would like to construct the space of functions that the integral is expected to evaluate to. This first requires ascertaining what types of special functions may appear. While one-loop Feynman integrals can be evaluated in terms of iterated integrals involving only $\d \log$ forms, more complicated classes of functions can appear at higher loops---for instance, functions involving integrals over both higher-genus curves~\cite{Huang:2013kh,Georgoudis:2015hca,Marzucca:2023gto} and higher-dimensional manifolds~\cite{Bloch:2014qca,Bourjaily:2018ycu,Bourjaily:2018yfy,Bourjaily:2019hmc,Duhr:2022pch,Duhr:2022dxb,Cao:2023tpx,McLeod:2023doa} are known to appear. 

To probe the types of functions that can arise in a given Feynman integral, we compute its leading singularities, as well as the leading singularities of the Feynman integrals that appear as its subtopologies (namely, the Feynman integrals one gets by contracting out all possible sets of propagators). For details on how these leading singularities can be computed, see Appendix~\ref{sec:maxcuts}. The crucial observation is that, in integrals that can be expressed in terms of iterated integrals involving only $\d\log$ integration kernels, it will be possible to find a sequence of residue contours that completely localize all integrations (both in the original Feynman integral, and the Feynman integrals that appear as its subtopologies). Conversely, in Feynman integrals that evaluate to more general classes of functions, such as elliptic multiple polylogarithms~\cite{2011arXiv1110.6917B,Broedel:2017kkb,Bogner:2019lfa}, one will not be able to localize all integrations via such sequences of residues. Rather, one will encounter integrals over nontrivial manifolds, such as elliptic curves. By identifying these obstructions to computing further residues, however, we learn what types of special functions may appear that go beyond $\d \log$ iterated integrals (see~\cite{Bourjaily:2021lnz,Bourjaily:2022bwx} for more background). This information can then be used to build an ansatz that draws upon the appropriate classes of special functions (see~\cite{Morales:2022csr} for an example involving elliptic curves).

Even in examples that are expected to be expressible in terms of iterated integrals over $\d \log$ forms, the leading singularities of a Feynman integral teach us about the rational and algebraic prefactors that will multiply these iterated integrals in the final result. When the number of propagators in an integral matches the number of integrations, this leading singularity is unique. However, when there are fewer propagators than integrations there will generally be several leading singularities. As an example, the triangle integral in two dimensions has three different leading singularities, each of which produces a different algebraic prefactor for the integral. After building these leading singularities into our ansatz as prefactors of the appropriate transcendental functions, we expect the remaining undetermined coefficients (which multiply these terms) to take rational numerical values. We will see this in examples below.

While in general, we expect that the Landau bootstrap method can be used to compute Feynman integrals that evaluate to general classes of special functions, we will here focus on examples that evaluate to iterated integrals involving only $\d \log$ forms. In these cases, the analytic structure of Feynman integrals can be described efficiently in terms of symbols~\cite{Gonch2,Goncharov:2010jf,Brown:2011ik,Brown1102.1312,Duhr:2012fh}, whose letters correspond to the algebraic functions (of external momenta and masses) that appear in $\d\log$ integration kernels. Thus, the locations where these letters vanish or become infinite correspond to logarithmic branch points, which we know can only arise on the singular hypersurfaces we have already cataloged.

In examples in which all singularities are expected to be logarithmic, a complete set of candidate symbol letters is already given by the polynomials that describe the singularities themselves. When algebraic singularities are also expected to appear, there is not a simple way to construct the alphabet without further information. For example, the integral may involve letters of the form $P \pm \sqrt{Q}$ where $P$ and $Q$ are polynomials in the external momenta. In these cases, the space of allowed $P$ and $Q$ is strongly restricted,  since any viable letter (and its inverse) should only vanish where singularities are expected to occur. 
In particular, such letters vanish on the locus $P^2 = Q$, so  $P^2 - Q$ must be a numeric multiple of (an integer power of) one (or more) of the Landau singularities. This allows us to construct an ansatz for the polynomials $P$ that can appear in symbol letters via the relation $P^2 = Q + C \prod_i \varphi_i^{c_i}$, where each factor $\varphi_i$ is a polynomial that vanishes on
a logarithmic singularity,  the $c_i$ are positive integers, and $C$ is a number~\cite{Heller:2019gkq}. It must then be the case that $Q + C \prod_i \varphi_i^{c_i}$ is the square of a polynomial.
This condition is a  powerful constraint.  
By scanning over all possible values of $Q$, $\varphi_i$, and $c_i$, one can systematically search for squares of polynomials.

In the general case, we should consider letters of the form
\begin{equation}
   L = \sum_i P_i Q_i^{\nu_i}  \, ,
   \label{eq:L_PQ}
\end{equation}
where the $P_i$ and $Q_i$ are polynomials in the external momenta and masses, and $\nu_i$ is a rational number.
For any possible form, one can explore the space of algebraic constraints similar to the $P\pm\sqrt{Q}$ case to constrain the possible letters. Once a set of candidate letters is determined, it can then be reduced to any subset that is multiplicatively independent, using tools from algebraic number theory (see for instance~\cite{Bourjaily:2019igt}).

\myparagraph{\bf Global Nature of Singularities}
Having constructed a candidate alphabet of symbol letters, we next constrain where each letter is allowed to appear in the symbol. The simplest set of constraints to impose come from \emph{$\alpha$-positivity}. Namely, while the symbol encodes all of the branch cuts that can be accessed as we analytically continue our Feynman integral, the singularities associated with the first entry must be encountered along the original (undeformed) contour of integration. Since the integration contour for the Feynman integral is over positive $\alpha_e$ and real values of $k_j^\mu$, the only branch points that should be accessible before the integral is analytically continued are those that arise for values $\smash{k_j^\mu \in \mathbb{R}^{1,\D-1}}$ and $\alpha_e \geq 0$.\footnote{To see this, one must first introduce Feynman parameters and recognize that they can be identified with the $\alpha_e$ parameters that appear in the Landau equations~\eqref{eq:Landau}.} We thus refer to such singularities as \emph{$\alpha$-positive singularities}.

The consequences of this observation differ for algebraic and logarithmic singularities. Consider for example the expressions
\begin{equation}
\frac{1}{\sqrt{x}} \log\frac{1-\sqrt{x}}{1+\sqrt{x}} \, , \qquad \log^2 \frac{1-\sqrt{x}}{1+\sqrt{x}} \,.
\label{xsform}
\end{equation}
These functions are not singular and do not develop any branch cuts as $x\to 0$ (on the principle sheet). Thus, it is possible for square roots that do not correspond to $\alpha$ positive singularities to appear in the first entry of the symbol, if they are appropriately compensated by square roots in prefactors or other symbol letters. More precisely, square root branch cuts will cancel out in any iterated integral that is invariant under the involution $\sqrt{\cdot} \to - \sqrt{\cdot}$\hspace{2pt}. As this mapping corresponds to one of the elements of the Galois group, we refer to iterated that have this property as being \emph{Galois even} under this sign flip.\footnote{More generally, in the case of a root of a degree $n$ polynomial, one would expect a Feynman integral to respect a total permutation symmetry with respect to all $n$ roots if the corresponding singularity is not $\alpha$-positive, as for instance seen in the results of~\cite{Chicherin:2024ifn}.}

Because of this Galois symmetry, algebraic singularities in one- and two-loop Feynman integrals can usually be put in the form
\begin{equation}
    L_i = \frac{P_i + \sqrt{Q_i}}{P_i - \sqrt{Q_i}} \, ,
    \label{eq:Li}
\end{equation}
where the $Q_i$ are products of singularities. Searching for letters of this form is vastly simpler than searching for letters of the more general form in equation~\eqref{eq:L_PQ}, so this represents a significant simplification when this assumption can be justified. 

Unlike algebraic singularities, logarithmic branch cuts that appear in the symbol cannot be canceled by algebraic prefactors. Therefore, logarithmic singularities that are not $\alpha$ positive \emph{cannot} appear in the first entry of the symbol.\footnote{In fact, the requirement that only $\alpha$-positive branch cuts are accessible in the physical region can also give rise to restrictions on the symbol beyond the first entry, as seen for example in~\cite{Dixon:2020bbt}. However, these further implications also follow from the asymptotics constrains that we describe below, so we won't work out these implications here.} On the other hand, logarithmic branch cuts that appear deeper in the symbol can be canceled off by beyond-the-symbol terms.
Indeed, logarithmic singularities which are not $\alpha$-positive generically do appear in subleading symbol entries and can be uncovered by taking discontinuities around other singularities. 

Further constraints can also be derived using knowledge of the types of behavior characterized in Eqns.~\eqref{eq:Landau_expansion} and \eqref{eq:genbehavior}. Namely, for Feynman integrals with generic masses, it has been shown that when the leading non-analytic behavior of $\I \sim \varphi^\gamma \log \varphi$ for $\gamma \in \mathbb{Z}^+$ near $\varphi=0$,
then $\varphi$ cannot appear in the 
the last $\gamma$ entries of the symbol~\cite{Hannesdottir:2021kpd}. 
A similar result can be derived for algebraic singularities:
when the leading non-analytic contribution of an integral near a square-root singularity behaves as in~\eqref{eq:Landau_expansion} with $\gamma \in \mathbb{Z}+\frac12$, no letter in the last $\gamma-\frac12$ symbol entries can involve square roots which vanish as $\varphi \to 0$. We show this in Appendix~\ref{app:asymptotic_algebraic}.

Although no such general results have been proven for integrals that do not involve generic masses, similar constraints can be derived on a case-by-case basis. 
Namely, once we know how a Feynman integral behaves locally around a branch point, we can simply expand our ansatz for $\I$ in this singular limit, and require that the ansatz exhibit the expected asymptotic behavior. As a simple example, if $\I \sim \varphi \log^2 \varphi$, the Feynman integral cannot evaluate to $\log^3 \varphi $ or $\varphi^2 \log \varphi $, but it could potentially evaluate to $\log^2 \varphi \, \Li_2 (\varphi)$. Knowledge of higher-order terms in the expansion of $\I$ around $\varphi \to 0$ can also be used to further restrict the coefficients in our ansatz, as we will see in the example of the massless box below.

\myparagraph{\bf Sequential Discontinuities}
Constraints can also be placed on the sequences of discontinuities that can be accessed in amplitudes and Feynman integrals. The first such class of constraints restrict which discontinuities can be taken in immediate succession; they are often referred to as \emph{adjacency constraints}, as they preclude certain pairs of symbol letters from appearing in adjacent entries of the symbol. A classic example is given by the Steinmann relations~\cite{Steinmann,Steinmann2,araki:1961,pham,Cahill:1973qp}, which forbid sequential discontinuities in partially overlapping channels in amplitudes that describe the scattering of stable massive particles. A number of generalizations of the Steinmann relations have been discovered in recent years~\cite{Drummond:2017ssj,Caron-Huot:2019bsq,Bourjaily:2020wvq,Benincasa:2020aoj}, which require more general sequences of discontinuities to vanish, either in individual Feynman integrals or scattering amplitudes. 

A second and more refined class of constraints can also be derived for Feynman integrals, which take into account the topology of the corresponding Feynman diagrams~\cite{pham,Hannesdottir:2022xki,Berghoff:2022mqu,Hannesdottir:2024cnn}. These types of constraints were pioneered by Pham~\cite{pham}, who used Picard-Lefschetz theory to rewrite the discontinuities of Feynman integrals in terms of integrals in which the integration contour has been localized to the on-shell locus of a subset of the propagators. As only some of the singularities that could be encountered in the original Feynman integral can still be accessed on the support of this on-shell space, only a subset of the discontinuities in the original Feynman integral will still exist. These restrictions go by the name of the \emph{hierarchical principle}~\cite{boyling1968homological,Landshoff1966}. More specifically, sequences of discontinuities can only be nonzero if the singular denominator factors that pinch the integration contour remain zero for all subsequent pinches~\cite{pham,Hannesdottir:2022xki}. While Pham's original approach required the Feynman integrals to have sufficiently generic masses, these types of constraints can more generally be derived after carrying out suitable blowups.

A hierarchical principle also exists in Feynman parameter space~\cite{Berghoff:2022mqu,Britto:2023rig}, where analogous topological constraints can be used to restrict which singularities can follow each other. In particular, a sufficient condition for the absence of a sequential singularity can be derived from the topology of the integration space, and can be automatized through computing its Euler characteristic. This approach, which leads to \emph{geneological constraints}~\cite{Hannesdottir:2024cnn}, proves to be quite powerful, as we will see in the example of the double box below.

\myparagraph{\bf Bootstrapping the Final Answer}

After deriving new constraints on the singularity structure of our chosen Feynman integral, we require that these constraints be obeyed by the ansatz we have constructed. In some cases, a unique answer may be found without it being necessary to carry out each of the steps highlighted above, while in other cases these classes of constraints may not yet be sufficient (or even possible to derive, using current technology). However, as long as a well-motivated ansatz can be constructed, it should always be possible to generate further constraints by expanding the Feynman integral to higher orders around its singular limits (as done, for instance, in~\cite{Chicherin:2017dob}). We now illustrate how the Landau bootstrap works in a number of examples.

\section{\large{Examples}}

\myparagraph{\bf Bubbles}

As a first example, we now illustrate the Landau bootstrap method for the example of the bubble integral
\begin{equation}
    \I_{\text{bub}}^{\D}(p) \!= \!\int \d^{\D} k \frac{1}{[(p-k)^2-m_1^2 + i \varepsilon][k^2-m_2^2 + i \varepsilon]} \, .
\end{equation}
We consider this integral for generic internal masses $m_1$ and $m_2$, so that we do not need to worry about carrying out blowups to identify singularities. This choice also allows us to predict the Landau exponent for each of the first-type singularities using equation~\eqref{eq:Landau_expansion}. We will bootstrap this integral in both $\D=2$ and $\D=3$ spacetime dimensions, in order to illustrate how the same types of constraints can lead to different functions. 

To carry out our bootstrap procedure, we begin by solving the Landau equations and characterizing the behavior we expect the bubble integral to exhibit near its singularities. The set of singularities we find are shown in Table~\ref{tab:bootstrap_bubble}. Note in particular that there is a second-type singularity for $D\ge 3$ at $s=0$, which we can detect by inverting the loop momenta as illustrated in Appendix~\ref{app:second-type} and in~\cite{Hannesdottir:2022bmo}. The threshold and pseudothreshold singularities at $s = (m_1 \pm m_2)^2$ give rise to square root branch cuts in $\D=2$ and logarithmic branch cuts in $\D=3$; conversely, the singularities at $m_i^2 = 0$ are logarithmic in $\D=2$ but algebraic in $\D=3$.

\begin{table}
\centering
    \begin{tabular}{c|c|c|c}
    Singularity & $\alpha$-positive? & \hspace{0.1cm} $\gamma$ in $\D=2$ \hspace{0.1cm} & \hspace{0.1cm} $\gamma$ in $\D=3$
    \hspace{0.1cm} \\ 
    \hline
    $s{-}(m_1 {+} m_2)^2$ & yes & ${-}\frac{1}{2}$ & $0$ \\
    $s{-}(m_1 {-} m_2)^2$ & no & ${-}\frac{1}{2}$ & $0$ \\
    $m_1^2$ & yes & $0$ & $\frac{1}{2}$ \\
    $m_2^2$ & yes & $0$ & $\frac{1}{2}$ \\[1mm]
    \hline
    $s$ & no & absent & $\substack{\text{not predicted} \\ \text{by equation~\eqref{eq:Landau_expansion}}}$ \\
    \end{tabular}
    \caption{Predictions for the singularities of the bubble integral in $\D=2$ and $\D=3$  dimensions and their Landau exponents $\gamma$ in the notation seen in~\eqref{eq:Landau_expansion}. The second-type singularity at $s=0$ can be shown to be absent in $\D=2$ and present but not $\alpha$-positive in $\D=3$. Its Landau exponent ($\gamma=-\frac{1}{2}$) in $\D=3$ cannot be predicted by equation~\eqref{eq:Landau_expansion}.
    } 
    \label{tab:bootstrap_bubble}
\end{table}

\vspace{.1cm}
\paragraph{\bf The Bubble in Two Dimensions}

We first consider the integral in $\D=2$, where there are logarithmic singularities at $m_i^2=0$ and algebraic singularities at $s=r_\pm^2 = (m_1 \pm m_2)^2$. This suggests that $m_1^2$ and $m_2^2$ are good candidate letters and that additional letters of the form seen in~\eqref{eq:Li} may arise, namely $\smash{P \pm R \sqrt{\vphantom{\vec{t}} \smash{(s - r_+^2)(s - r_-^2)}}}$ for some polynomials $P$ and $R$. Multiplying these by $\sqrt{\vphantom{\vec{t}} \smash{s-r_+^2}}$, we get letters of the form
$\smash{P_+ \sqrt{\vphantom{\vec{t}} \smash{s - r_+^2}} + P_- \sqrt{\vphantom{\vec{t}} \smash{s-r_-^2}}}$ for some polynomials $P_\pm(s,m_1,m_s)$. In order not to introduce logarithmic singularities at unwanted locations, we require that $(s-r_+^2)P_-^2 - (s-r_-^2)P_+^2 =0$ is only satisfied when $m_1=0$ or $m_2=0$. This is possible only if $P_+^2 =P_-^2 =1$ (see Appendix~\ref{sec:symbols-from-sqrt}).
We are thus led to the four candidate letters
\begin{align} \label{eq:bubble_2d_letters}
\{A_1,A_2, A_\pm \} = \left\{m_1^2\, ,m_2^2\, , \sqrt{s-r_+^2} \pm \sqrt{s-r_-^2} \right\} \, .
\end{align}
Note that $A_+ A_- = 4 m_1 m_2$, so we could choose to drop one of the letters from our ansatz. Instead, we choose a basis consisting of $A_1$, $A_2$ and $A_+/A_-$, which slightly simplifies the analysis below.

Next, we note that sequential discontinuities in the two masses are forbidden by the hierarchical principle:
\begin{align}
  \label{eq:double-disc-mass}
\text{Disc}_{m_2^2} \text{Disc}_{m_1^2} \, (\I_{\text{bub}}) = \text{Disc}_{m_1^2} \text{Disc}_{m_2^2} \, (\I_{\text{bub}}) = 0. 
\end{align}
Additionally, the final result cannot involve $\log^n m_i$ with $n>1$, since the Landau exponent for the masses is $\gamma=0$. As $A_\pm$ also give rise to logarithmic branch points when $m_i^2=0$, the same conclusions hold for more complicated polylogarithms that also involve these letters. The upshot is that the bubble integral in $\D=2$ must be expressible in terms of the ansatz
\begin{align}
\I_{\text{bub}}^{\D=2} = c_1 \log A_1 + c_2 \log A_2 + c_{\pm} \log \frac{A_+}{A_-} \, , \label{eq:2d_bubble_ansatz}
\end{align}
where the coefficients $c_1$, $c_2$, and $c_{\pm}$ are expected to be algebraic functions of $s$, $m_1$, and $m_2$. This is consistent with a general result that the maximum transcendental weight of an $\ell$-loop integral in $\D$ dimensions with generic masses is $\floor{\frac{\ell D}{2}}$~\cite{Hannesdottir:2021kpd}.

To fix the value of the coefficients in~\eqref{eq:2d_bubble_ansatz}, we next compute the leading singularity of the bubble integral in $\D=2$ as discussed above, finding 
\begin{align}
\text{LS}\left(\I_{\text{bub}}^{\D=2}\right) =\frac{4 \pi^2}{\sqrt{s-r_+^2}\sqrt{s-r_-^2}}\,.
\end{align}
Since this integral has a unique leading singularity, we expect each of the coefficients $c_i$, $i \in \{1,2, \pm\}$, to be a rational multiple of this function. This means, in particular, that the threshold and pseudothreshold singularities will appear as square roots in the prefactor of each term in our ansatz in equation~\eqref{eq:2d_bubble_ansatz}. As the pseudothreshold is not an $\alpha$-positive singularity, this implies that the result must be invariant under the Galois symmetry $\smash{\sqrt{\vphantom{\vec{t}} \smash{s - r_-^2}} \to -\sqrt{\vphantom{\vec{t}} \smash{s - r_-^2}}}$. The last term multiplying $c_{\pm}$ already has this property, as it maps back to itself when the sign in front of the pseudothreshold square root is flipped. As a consequence, we also have $c_1=c_2=0$, since there are no additional square roots in these other two terms that can cancel off the pseudothreshold branch cut in the physical region. This fixes the result to be
\begin{align}
\I_{\text{bub}}^{\D=2} & = \frac{2 i \pi}{\sqrt{s-r_+^2}\sqrt{s-r_-^2}} \log \left(\frac{\sqrt{s-r_+^2} - \sqrt{s-r_-^2}}{\sqrt{s-r_+^2} + \sqrt{s-r_-^2}} \right) \nonumber\\
&=\frac{i \pi}{\sqrt{r_{12}}} \log\frac{(s-m_1^2-m_2^2 -\sqrt{r_{12}})^2}{4 m_1^2 m_2^2}
\end{align}
with $r_{12}  = s^2-(m_1^2-m_2^2)^2-2s(m_1^2+m_2^2)$. 
The second form (which depends only on $m_i^2$, not $m_i$) manifests that there is no algebraic singularity in $m_i^2$, and is valid when $s<(m_1+m_2)^2$.
The numerical proportionality factor can be determined by comparing the discontinuity across the cut $s = r_+^2$ with Cutkosky's formula, which yields the leading singularity. The branch of the logarithm is determined by requiring that $\I_{\text{bub}}^{\D=2}$ is free of branch cuts for $s<(m_1+m_2)^2$. 

\vspace{.1cm}
\paragraph{\bf The Bubble in Three Dimensions}
Let us now see how the same methods lead to a different result in $\D=3$. Referring back to Table~\ref{tab:bootstrap_bubble}, we observe that a second-type singularity can now arise at $s=0$. In fact, this singularity already appears in the leading singularity in $\D=3$: 
\begin{align} \label{eq:ls_3d}
\text{LS}\left(\I_{\text{bub}}^{\D=3}\right) = \frac{2 \pi^3}{\sqrt{s}}\,.
\end{align}
Since this singularity is not $\alpha$-positive, the transcendental function that this overall prefactor multiplies must be odd under the map $\sqrt{s} \to - \sqrt{s}$. For the same reason, we know that a logarithmic branch cut cannot arise at $s=0$ in the first entry of the symbol. However, as the transcendental weight of this integral is bounded to be no more than one~\cite{Hannesdottir:2021kpd}, this implies that no logarithmic branch cuts arise at $s=0$.  

Combining this information about the second-type singularity with our other expectations from Table~\ref{tab:bootstrap_bubble}, that we only get algebraic branch cuts at $m_i^2=0$ and logarithmic branch cuts at $s = (m_1\pm m_2)^2$, we find we can generate four independent candidate symbol letters:
\begin{equation}
        s+m_1^2-m_2^2 \pm 2 \sqrt{m_1^2} \sqrt{s} \,, \quad
    s+m_2^2-m_1^2 \pm 2 \sqrt{m_2^2} \sqrt{s}\, .
\end{equation}
Only certain combinations of these letters are consistent with the $\alpha$-positivity predictions for this integral, which dictate that only the threshold singularity is accessible in the physical region:
\begin{align}
    A_1 & = \frac{-s-m_1^2+m_2^2 + 2 \sqrt{m_1^2} \sqrt{s}}{-s-m_1^2+m_2^2 - 2 \sqrt{m_1^2} \sqrt{s}} \,, \\
    A_2 & =  \frac{-s-m_2^2+m_1^2 + 2 \sqrt{m_2^2} \sqrt{s}}{-s-m_2^2+m_1^2 - 2 \sqrt{m_2^2} \sqrt{s}} \,.
\end{align}
By symmetry in $m_1 \leftrightarrow m_2$, the only possible symbol letter becomes
\begin{equation}
    A_1 A_2 = \left(\frac{\sqrt{m_1^2} + \sqrt{m_2^2} - \sqrt{s}}{\sqrt{m_1^2} + \sqrt{m_2^2} + \sqrt{s}} \right)^2 \,,
\end{equation}
and the final answer must be
\begin{align}
\I_{\text{bub}}^{\D=3} = \frac{i \pi^2}{\sqrt{s}} \log \left(\frac{\sqrt{m_1^2} + \sqrt{m_2^2} - \sqrt{s}}{\sqrt{m_1^2} + \sqrt{m_2^2} + \sqrt{s}} \right) \, ,
\label{d3bubble}
\end{align}
where the proportionality constant can be fixed as in the two-dimensional case.  We note that in this case the relation in equation~\eqref{eq:double-disc-mass} is nontrivially satisfied.

\myparagraph{\bf The Massless Box in Six Dimensions}
Next, we consider the box integral in six dimensions with massless internal and external particles:
\begin{equation}
    \begin{gathered}
    \raisebox{0pt}[\height][\depth]{\hspace{-30pt}%
    \begin{tikzpicture}[scale=0.4, thick]
    \draw [decorate,decoration={brace,amplitude=5pt,mirror}]
    (-2.5,1.5) -- (-2.5,-1.5);
    \node at (-4.5,0) {$s$};
    \draw [decorate,decoration={brace,amplitude=5pt,mirror}]
    (-1.5,-2.5) -- (1.5,-2.5);
    \node at (0,-4) {$t$};
    \draw[dashed] (-1,1) -- (1,1) -- (1,-1) -- (-1,-1) --cycle;
    \draw[dashed] (-1,1) -- ++(135:1);
    \draw[dashed] (1,1) -- ++(45:1);
    \draw[dashed] (1,-1) -- ++(-45:1);
    \draw[dashed] (-1,-1) -- ++(-135:1);
    \end{tikzpicture}
    }
    \end{gathered}
\end{equation}
This integral is finite, and can be evaluated (for instance, via direct integration) to give
\begin{equation}
    \I_{\text{box}}^{\D=6} = - i \pi^3 \frac{\log^2\left(\frac{-s-i\varepsilon}{-t-i\varepsilon}\right) + \pi^2}{2 (s+t)} \,. \label{eq:6d_box}
\end{equation}
We now show how we can go about deriving the same result using the Landau bootstrap. Going forward, we assume that $s,t<0$ so we can drop the $i\varepsilon$'s.

We again begin by solving the Landau equations, finding solutions when
\begin{equation}
    s, t, u \in \{0,\infty\} \,,
    \label{eq:stu_sings}
\end{equation}
where $u=-s-t$.  We would also like to predict what types of branch cuts can arise at these points; however, we have checked that the singularities at $s=0$ and $t=0$ do not arise from simple pinches (using the methods for expanding $\I_{\text{box}}^{\D=6}$ described in~\cite{Landau:1959fi,pham2011singularities}). This means that the easiest way to learn about these limits is using the method of regions. In particular, this integral can be computed as an expansion around the $\frac st \to 0$ limit at fixed $t$, whereupon it is found it takes the form
\begin{align}
        \label{eq:lims0}
        \mathcal{I}^{\D=6}_{\text{box}}\left(\frac st\to 0 \right) & = - i \pi^3 \Bigg[ \frac{\log^2 (s/t)}{2 t} \sum_{i=0}^\infty g_i \Big(\frac{s}{t}\Big)^i  \\
        \nonumber & \qquad 
        + \frac{\pi^2}{2 t} \sum_{i=0}^\infty h_i \Big(\frac{s}{t}\Big)^i \Bigg] \,,
\end{align}
where $g_i$ and $h_i$ are numbers. The limit as $\frac{s}{t} \to \infty$ can be obtained by symmetry as the $\frac{t}{s} \to 0$ limit. Importantly, these results imply that the branch cuts that develop in these limits will only be logarithmic. 

In fact, we can derive the same result starting from the general ansatz we wrote for symbol letters in equation~\eqref{eq:L_PQ}. Given that the only singularities in this integral appear when one of $s$, $t$, or $u$ vanish or becomes infinite, we must have that each $Q_i$ and $P_i$ take the form $s^a t^b u^c$ for some integers $a$, $b$, and $c$. It is not hard to convince oneself that constructing an algebraic symbol letter out of these monomials will always introduce singularities that go beyond the loci in equation~\eqref{eq:stu_sings}. We conclude that there can be no algebraic symbol letters in this Feynman integral.\footnote{This second argument leaves open the possibility that the prefactor of the integral involves a square root; however, this possibility can be ruled out by computing the leading singularity of this integral, which is rational.}

Only two multiplicatively-independent symbol letters can be constructed whose logarithmic singularities are contained within the set equation~\eqref{eq:stu_sings}. We choose a basis of letters given by
\begin{equation}
    B_1 = \frac{s}{t} \equiv x \,, \qquad B_2 = - \frac{u}{t} \equiv 1+x \,,
\end{equation}
in terms of which we can construct an ansatz for the symbol of the box integral:
\begin{align}
    &\sum_{i,j,k=1}^2 c_{ijk} B_i \otimes B_j \otimes B_k + \sum_{i,j}^2 c_{ij} B_i \otimes B_j
    + \sum_{i=1}^2 c_i B_i\,. \label{eq:box_ansatz}
\end{align}
Here, the coefficients $c_\bullet$ are allowed to be functions of $s$, $t$, and $u$. Note that this ansatz is automatically integrable, since it only depends on a single variable $x$.

To constrain the coefficients that appear in~\eqref{eq:box_ansatz}, we first require that it is invariant under the exchange $s \leftrightarrow t$, which sends $x \to x^{-1}$ and $1+x \to \frac{1+x}{x}$. Moreover, using the fact that the only $\alpha$-positive solutions to the Landau equations correspond to $x = \{0,\infty\}$, we disallow the letter $B_2$ from appearing in the first entry. With these constraints imposed, our ansatz reduces to
\begin{multline}
  c_{112} x \otimes x \otimes (1 + x) +
  c_{121} x \otimes (1 + x) \otimes x \\
  -\frac 1 2 (c_{111} + c_{121}) x \otimes x \otimes x
  + c_{11} x \otimes x.
\end{multline}
We now require that the expansion of this symbol
as $\frac{s}{t} \to 0$ and $\frac{t}{s} \to 0$ matches the known form of the expansion in equation~\eqref{eq:lims0}. This imposes the requirement that $c_{112} = c_{121} = 0$ and sets the overall scale of our ansatz, giving us the unique result
\begin{align}
    \S (\I_{\text{box}}^{\D=6}) \propto \frac{1}{s+t} \left( \frac{s}{t} \otimes \frac{s}{t} \right) \,.
\end{align}
Finally, we can upgrade this result to full function level, again using the fact that the singularity at $u = -s-t$ is not $\alpha$ positive. This means that we need to choose a constant $c_0$ such that the expression
\begin{equation}
   \I_{\text{box}}^{\D=6} \propto  \frac{\log^2\left(\frac{s}{t}\right)}{2 (s+t)} + c_0 
\end{equation}
remains finite when $s=-t$. This gives us the constraint that $(\pm i \pi)^2 + 2(s+t)c_0 = 0$. Our final result for the box integral thus becomes
\begin{equation}
    \I_{\text{box}}^{\D=6} \propto \frac{\log^2\left(\frac{s}{t}\right)+ \pi^2}{2 (s+t)} \,,
\end{equation}
which matches the expression in~\eqref{eq:6d_box} after fixing the prefactor using~\eqref{eq:lims0}.

\myparagraph{\bf The Double Box with a Massive Internal Loop}
As a final example, we consider the double box integral in which the outermost loop of propagators has been given a mass:
\begin{equation}
    \begin{gathered}
    \raisebox{0pt}[\height][\depth]{\hspace{-30pt}%
    \begin{tikzpicture}[scale=0.5, thick]
    \draw [decorate,decoration={brace,amplitude=5pt,mirror}]
    (-2.5,1.8) -- (-2.5,-1.8);
    \node at (-4.5,0) {$s$};
    \draw [decorate,decoration={brace,amplitude=5pt,mirror}]
    (-1.8,-2.5) -- (3.8,-2.5);
    \node at (1,-4) {$t$};
    \draw[line width=1.2] (-1,1) -- (3,1) -- (3,-1) -- (-1,-1) --cycle;
    \node[] at (0,1.8) {$m$};
    \draw[dashed] (1,1) -- (1,-1);
    \draw[dashed] (-1,1) -- ++(135:1.5);
    \draw[dashed] (3,1) -- ++(45:1.5);
    \draw[dashed] (3,-1) -- ++(-45:1.5);
    \draw[dashed] (-1,-1) -- ++(-135:1.5);
    \end{tikzpicture}
    }
    \end{gathered}
\end{equation}
Here, all of the solid lines denote propagators of mass $m$, while dashed lines denote massless particles.
We consider this integral in $\D=4$, where it is both UV and IR finite. In~\cite{Caron-Huot:2014lda}, this integral was evaluated in terms of polylogarithms in $\D=4-2\epsilon$. 

We first catalog all the solutions to the Landau equations that can be found using current methods. Our survey turns up $\alpha$-positive singularities at
\begin{equation}
\begin{gathered}
    s= 4m^2, \quad s\to \infty, \\ t=4m^2, \quad t\to \infty,  \\ m^2= 0, \label{eq:thresholds}
\end{gathered}
\end{equation}
and additional singularities at
\begin{equation}
\begin{gathered}
    s=0, \qquad t=0, \qquad m^2 \to \infty \,,   
 \\
   s+t = 0, \qquad  st - 4 m^2 s - 4 m^2 t = 0 \,.
    \label{eq:zerosols}
\end{gathered}
\end{equation}
This integral also has a singularity at $s t^2-2 s t m^2+s m^4-4 t^2 m^2=0$ at higher orders in $\epsilon$; however, we show in Appendix~\ref{app:second-type} that the corresponding solution to the Landau equations only appears outside of four dimensions. Since we are bootstrapping the double box in strictly four dimensions, we do not include this singularity in our calculation.

Adopting the conventions of~\cite{Caron-Huot:2014lda}, we formulate the two dimensionless variables that the transcendental part of this integral can depend on as
\begin{equation}
    u = - \frac{4m^2}{s}\,, \qquad v = - \frac{4m^2}{t} \,.
\end{equation}
In these variables, the solutions to the Landau equations cataloged above correspond to either $u$, $v$, or $u+v$ taking one of the values $\{-1,0,\infty\}$. Since, in general, each of these singularities arises from divergences that occur at multiple locations along the integration contour (where each of these pinches may have a different Landau exponent), we will assume that both logarithmic and square root branch cuts can arise on each of these kinematic hypersurfaces. The exception is that we assume that the square roots $\sqrt{u}$ and $\sqrt{v}$ cannot arise in symbol letters, since their presence would lead to a different Regge limit than the one computed using the method of regions~\cite{PhysRev.131.480}. In particular, the limits as $u \to 0$ (at fixed $v$) and $v \to 0$ (at fixed $u$) show that the singularities at $u=0$ and $v=0$ are logarithmic in $\D=4$ spacetime dimensions.

Carrying out a systematic search for letters that can be built out of these singularities, we find only twelve multiplicatively-independent letters. These letters are presented in Table~\ref{tab:double_box_letters}, 
where (again following~\cite{Caron-Huot:2014lda}) we have defined
\begin{gather}
\beta_{u} = \sqrt{1+u}\,, \qquad \beta_{v} = \sqrt{1+v}\, , \\  \beta_{uv} = \sqrt{1+u+v}\,.
\end{gather}
Letters $L_1$ through $L_{10}$ match the letters that are known to appear in the (dimensionally-regulated) answer~\cite{Caron-Huot:2014lda}; however, we find two additional letters $L_{11}$ and $L_{12}$ that cannot be ruled out simply because of where they give rise to branch cuts. We correspondingly include them in our bootstrap calculation.

\begin{table}
\renewcommand{\arraystretch}{1.5}
\begin{center}
\begin{tabular}{c | c c c} Letter \ & \hspace{0.0cm} Definition \hspace{0.0cm} & $\log(\bullet)$ sing. & $\sqrt{\bullet}$ sing. \\
\hline
 $L_1$ & $u$ & $s$, $m^2$, $\frac{1}{s}$, $\frac{1}{m^2}$ & -  \\ 
 $L_2$ & $v$ & $t$, $m^2$, $\frac{1}{t}$, $\frac{1}{m^2}$ & - \\  
 $L_3$ & $1+u$ & 
 $s$, $s\!-\!4m^2$,  $\frac{1}{m^2}$ & - \\ 
 $L_4$ & $1+v$ & $t$, $t\!-\!4m^2$,   $\frac{1}{m^2}$ & - \\ 
 $L_5$ &$u+v$ &  $s$, $t$, $s+t$, $m^2$,  $\frac{1}{m^2}$ & - \\ 
 $L_6$ & $\frac{\beta_u-1}{\beta_u+1}$ &  $m^2$, $\frac{1}{s}$& $s-4m^2$, $s$, $\frac{1}{m^2}$ \\ 
 $L_7$ & $\frac{\beta_v-1}{\beta_v+1}$ & $m^2$, $\frac{1}{t}$  & $t-4m^2$, $t$, $\frac{1}{m^2}$ \\
 $L_8$ & $\frac{\beta_{uv}-1}{\beta_{uv}+1}$ & $s+t$, $m^2$ & \begin{tabular}{c} $s$, $t$, \\[-5pt] $st{-}4m^2 s{-}4m^2 t$ \end{tabular}  \\   
 $L_9$ & $\frac{\beta_{uv}-\beta_u}{\beta_{uv}+\beta_u}$ & $s$, $m^2$, $\frac{1}{t}$ & \begin{tabular}{c} $s-4m^2$, $t$, \\[-5pt] $st{-}4m^2 s{-}4m^2 t$ \end{tabular}  \\ 
 $L_{10}$ & $\frac{\beta_{uv}-\beta_v}{\beta_{uv}+\beta_v}$ & $t$, $m^2$, $\frac{1}{s}$  & \begin{tabular}{c} $t-4m^2$, $s$, \\[-5pt] $st{-}4m^2 s{-}4m^2 t$ \end{tabular}  \\ 
 $L_{11}$ & $1+u+v$ & 
 \begin{tabular}{c}
 $s$, $t$, $\frac{1}{m^2}$, \\[-5pt] $st{-}4m^2 s{-}4m^2 t$ 
 \end{tabular} & -
 \\ 
  $L_{12}$ & $\frac{\beta_{uv}-\beta_u \beta_v}{\beta_{uv}+\beta_u \beta_v}$ \!\!\!\!\!\! & $\frac{1}{s}$, $\frac{1}{t}$, $m^2$ &  \begin{tabular}{c}
$s-4 m^2$, $t-4 m^2$, \\[-5pt] \!\!\!\! $m^{-2}$, $st{-}4m^2 s{-}4m^2 t$
 \end{tabular}
\\ 
\end{tabular}
\end{center}
\caption{The twelve symbol letters that can be constructed out of the solutions to the Landau equations in equations~\eqref{eq:thresholds} and~\eqref{eq:zerosols}, and the kinematic loci where each of them develops logarithmic and square root branch cuts. Although a priori we allow for any of these letters, $L_{11}$ and $L_{12}$ happen not to appear in the final answer.}
\label{tab:double_box_letters}
\end{table}

In order to formulate an ansatz for $\I^{\D=4}_{\text{dbox}}$, we next determine the prefactor of this integral. Computing its leading singularity, we find
\begin{equation}
\text{MaxCut}\left(\I^{\D=4}_{\text{dbox}}\right) \propto \tfrac{1}{s^2 t \sqrt{1 + u} \sqrt{1 + u + v}} \, .
\end{equation}
Correspondingly, we take our initial ansatz for the symbol of $\I^{\D=4}_{\text{dbox}}$ to be
\begin{gather}
     \S (\I^{\D=4}_{\text{dbox}})  \propto \tfrac{1}{s^2 t \sqrt{1 + u} \sqrt{1 + u + v}} \, \tilde{\I}_{\text{dbox}}\, , \label{eq:db_box_ansatz}
     \\
     \tilde{\I}_{\text{dbox}} = \sum c_{i_1,i_2,i_3,i_4} L_{i_1} \otimes L_{i_2} \otimes L_{i_3} \otimes L_{i_4}\,,
     \label{eq:symI}
\end{gather}
where the coefficients $c_{i_1,i_2,i_3,i_4}$ are rational numbers, and the sum ranges over all values of $\{i_1,i_2,i_3,i_4\} \in \{1, \dots, 12\}$. Note that this ansatz assumes $\I^{\D=4}_{\text{dbox}}$ will evaluate to a polylogarithm with uniform transcendental weight four; this expectation comes from the fact that this integral is (proportional to) a pure integral in the sense of~\cite{Arkani-Hamed:2010pyv}.

We next require our ansatz to be integrable (see Appendix~\ref{app:integrability}).
This reduces the number of free coefficients to 6993. In addition, since no multiplicative combination of the square roots $\beta_u$, $\beta_v$, and $\beta_{uv}$ can be formed such that all non-$\alpha$-positive branch cuts cancel, we conclude that our ansatz must be Galois even with respect to flipping the sign in front of each of these roots. Taking into account the algebraic prefactor in equation~\eqref{eq:db_box_ansatz}, this means each term in the symbol of $\tilde{\I}_{\text{dbox}}$ must involve an even number of letters that depend on $\beta_v$ and an odd number of letters that depend on each $\beta_u$ and $\beta_{uv}$. Imposing this constraint, we further reduce the number of free coefficients to 861. 

We must separately impose $\alpha$-positivity constraints on the logarithmic branch cuts that appear in the first entry of the symbol. Only 7 letters can be formed that are free of non-$\alpha$-positive logarithmic singularities: 
\begin{align}
\left\{\frac{L_1}{L_3},\, \frac{L_2}{L_4} ,\, \frac{L_3 L_8}{L_5 L_{10}} ,\, \frac{L_4 L_8}{L_5 L_{9}} ,\, L_6 ,\, L_7 ,\, L_{12} \right\}
\end{align}
We thus allow only these letters to appear in the first entry of the symbol.

Next, we leverage the hierarchical principle. While deducing all of the implications of this principle (which would require identifying all solutions to the Landau equations) remains hard, some of these constraints can be easily deduced using the genealogical methods introduced in~\cite{Hannesdottir:2024cnn}. Using that technology, we find that threshold discontinuities at $t = 4m^2$ cannot follow discontinuities with respect to any of the branch cuts in the $s$ plane. This implies that the symbol letter $L_4$ cannot appear after $L_1$, $L_3$, $L_5$, $L_6$, $L_8$, $L_9$, $L_{10}$, $L_{11}$, or $L_{12}$, as all of the latter give rise to logarithmic branch points when $s$ is equal to $0$, $4m^2$, $-t$, $\frac{4 m^2 t}{t - 4 m^2}$, or $\infty$. It also imposes further restrictions on letters in which algebraic branch cuts arise at $t=4m^2$; for instance, it implies that $L_{12}$ cannot appear in the first entry, as such a first entry would allow one to compute a logarithmic discontinuity with respect to the branch point at $s \to \infty$, followed by an algebraic discontinuity at $t = 4m^2$ (since our Galois constraints ensure that an odd number of the letters appearing after $L_{12}$ involve $\beta_v$). Similarly, $L_{7}, L_{10}$, and $L_{12}$ cannot appear after the above list of letters with logarithmic branch points in $s$.
After imposing these genealogical constraints, we find our ansatz has only 28 free coefficients remaining.

We can fix most of these final 28
coefficients by leveraging more refined information about the solutions to the Landau equations that identify singularities at $s=4m^2$. Only four of these solutions are $\alpha$-positive---two which correspond to bubble Landau diagrams in which two massive lines have been cut, and two which correspond to sunrise Landau diagrams in which one massless and two massive lines have been cut. The Landau exponent associated with the massive bubble singularities is given by equation~\eqref{eq:Landau_expansion} as $\frac 1 2$, while the Landau exponent associated with the sunrise singularities also turns out to be $\frac 1 2$~\cite{unpublished}. As such, we do not expect $L_3$, which has a logarithmic singularity at $s=4m^2$, to appear in the first entry. Imposing this requirement on our ansatz reduces the number of free coefficients to just 6. 

\begin{table}
    \centering
    \begin{tabular}{|l|c|}
    \hline
    Sequences of four letters \ \ & \ \ 20736 \ \ \\
    Integrable weight-four symbols \ \ & \ \ 6993 \ \ \\
    Galois symmetry & 861 \\
    Physical logarithmic branch cuts & 161 \\
    Genealogical constraints & 28 \\
    Only algebraic $\alpha$-positive thresholds \ \ & 6 \\
    Threshold expansion in $t$ & 1\\
    \hline
    \end{tabular}
    \caption{The number of free parameters that remain after imposing successive constraints on our ansatz for $\S (\I^{\D=4}_{\text{dbox}})$.}
    \label{tab:bootstrap_dbox}
\end{table}

Finally, using the method of regions, we can show that the expansion of this integral around the $t=4m^2$ threshold involves square root branch points, but no large logarithms~\cite{HS}. This allows us to rule out $L_4$ from appearing anywhere in our ansatz, which fixes one more of our free coefficients. Determining the value of the remaining overall coefficient by comparing to the leading singularity, we are left with the symbol
\begin{align}
\label{eq:symbol}
\mathcal{S}(\tilde{\I}_{\text{dbox}}) & =
- L_6 \otimes \frac{L_1}{L_3} \otimes L_6 \otimes L_9 - L_6 \otimes \frac{L_1}{L_3} \otimes L_9 \otimes L_6 
\nonumber
\\ &
+ L_6 \otimes L_6 \otimes \frac{L_1 L_2}{L_3 L_5} \otimes L_9
+ L_6 \otimes L_9 \otimes \frac{L_2}{L_5} \otimes L_6 
\nonumber
\\ &
+ L_6 \otimes L_6 \otimes L_8 \otimes L_6  
+ L_6 \otimes L_9 \otimes L_8 \otimes L_9
\nonumber
\\ &
+ L_7 \otimes L_{10} \otimes \frac{L_2}{L_5}  \otimes L_6
+ L_7 \otimes L_{10} \otimes L_8 \otimes L_9
\nonumber
\\ &
+ L_7 \otimes L_7 \otimes \frac{L_1}{L_5} \otimes  L_9
+ L_7 \otimes L_7 \otimes L_8 \otimes L_6 \, .
\end{align}
This symbol matches the $\mathcal{O}(\epsilon^0)$ contribution to what was found in~\cite{Caron-Huot:2014lda}.

\myparagraph{\bf Conclusions}
The extent to which the analytic properties of Feynman integrals dictate their functional form is remarkable, as is the extent to which these analytic properties can be predicted. We have here presented a new Landau bootstrap method, which attempts to make practical use of these observations by 
determining Feynman integrals from only knowledge of their singular behavior. This strategy involves determining both the locations and nature of the singularities that appear in a given Feynman integral, in order to construct and constrain an ansatz that has these properties built in.

As a proof of principle, we have shown that the Landau bootstrap method can be used to compute the double-box integral with a massive internal loop, reproducing the results of~\cite{Caron-Huot:2014lda}. In doing so, we have only made use of information about the singular behavior of this integral that can be predicted using current technology. Given the flurry of recent work on understanding the singularity structure of Feynman integrals, we expect that our ability to make these types of predictions will continue to grow. In more difficult examples, we also have the ability to incorporate constraints that come from special kinematics limits around which Feynman integrals can be more easily computed, such as soft, collinear, and Regge limits. (In principle, Feynman integrals can be computed as an expansion around any of their singular points.) All of these inputs can help fix the coefficients in a well-motivated ansatz. 

More generally, we believe that substantial advancements can and soon will be made in our ability to execute each of the steps that comprise the Landau bootstrap. This makes us optimistic that it will soon be possible to compute a wide range of Feynman integrals---and by extension, scattering amplitudes---from knowledge of their singularity properties alone.

\vspace{.4cm}
\acknowledgments%
CV wishes to thank Jacob Bourjaily and John Collins for discussions. AM and CV are grateful to the Simons Center for Geometry and Physics for hospitality.   AJM is supported by the Royal Society grant URF{\textbackslash}R1{\textbackslash}221233, CV is supported in part by grant 00025445 from Villum Fonden. HSH gratefully acknowledges funding provided by the J. Robert Oppenheimer Endowed Fund of the Institute for Advanced Study. 
HSH and MDS are supported in part by
the U.S. Department of Energy, Office of Science, Office of High Energy Physics under Awards DE-SC0009988 and DE-SC0013607 respectively.
\appendix

\setcounter{secnumdepth}{2}

\section{Asymptotic Constraints on Algebraic Branch Points}
\label{app:asymptotic_algebraic}
In this appendix we derive a connection between the leading power at which square root branch points appear in Feynman integrals, and the positions in the symbol where the corresponding square roots appear. This derivation parallels the one given for logarithmic branch points in~\cite{Hannesdottir:2021kpd}, and similarly assumes that the Feynman integrals under consideration involve generic internal and external masses. At the end of the section, we comment on how this result can change when not all masses are generic, and describe how this type of asymptotic behavior can still be leveraged to bootstrap Feynman integrals with particles of equal or vanishing mass.

We start by considering a polylogarithmic Feynman integral $\I$ whose symbol can be abstractly denoted as
\begin{equation} \label{eq:symb_Feynman_integral}
    \mathcal{S}(\I) = \sum a_1 \otimes \dots \otimes a_n \, ,
\end{equation}
where we have left the kinematic dependence of the integral $\I$ and each of the symbol letters $a_j$ implicit. We have also left the sum in~\eqref{eq:symb_Feynman_integral} schematic, since we will analyze it term by term. 

We can construct the iterated integral corresponding to each term in~\eqref{eq:symb_Feynman_integral} by pulling back each $d \log a_j$ to an auxiliary space of integration variables $t_i$. For simplicity, we choose the integration contour to be a straight path 
\begin{equation}
    \sigma_i(t) = (1-t) p_i^\bullet  + t p_i 
    \label{eq:sigma}
\end{equation}
that interpolates between some generic initial set of kinematic values $p_i^\bullet$ and the kinematic point $p_i$. The pullback of the $d\log$s are then given by
\begin{equation}
    \sigma^* ( d \log a_j) (t) = \frac{(p_i - p_i^\bullet) \cdot (\nabla_i a_j )(\sigma(t)) }{a_j(\sigma(t))} dt \, ,
    \label{eq:pullback}
\end{equation}
in terms of which we can rewrite the symbol of $\I$ from~\eqref{eq:symb_Feynman_integral} as an iterated integral
\begin{align} 
    \I & = \sum \int_{0 \leq t_1 \leq \cdots \leq t_n \leq 1} \sigma^* (d \log a_1)(t_1) 
    \label{eq:iterated_integral_Feynman_integral}
     \\ & \qquad \times \sigma^* (d \log a_2)(t_2)
    \dots \sigma^* (d \log a_n)(t_n) \, .
    \nonumber
\end{align}
Note that the equality in~\eqref{eq:iterated_integral_Feynman_integral} only holds up to transcendental constants, since our integration contour has been chosen to start at a generic point. 

We now focus on the contribution from a single term in the sum~\eqref{eq:iterated_integral_Feynman_integral}, in which we assume the letter $a_m$ depends on $\sqrt{\varphi}$. By changing the order of integration, we can rewrite the iterated integral in this term as
\begin{equation} \label{eq:iterated_integral_different_form}
    \int_0^1 U(t)  \sigma^*(d \log a_m(\sqrt{\varphi})) (t) V(t) \, ,
\end{equation}
where $U(t)$ and $V(t)$ are themselves iterated integrals defined by
\begin{equation}
U(t) = \int_0^t \sigma^*(d \log a_1)(t_1) \cdots \int_{t_{m - 2}}^t \!\!\sigma^*(d \log a_{m-1})(t_{m-1}) \,  
\end{equation}
and
\begin{equation} \label{eq:V_def}
V(t) = \int_t^1 \sigma^*(d \log a_{m+1})(t_{m+1}) \cdots \int_{t_{n - 1}}^1 \!\! \sigma^*(d \log a_n)(t_n) \, .
\end{equation}
We assume for now that $\sqrt{\varphi}$ does not appear in any of the letters after position $m$, so $V(t)$ depends only rationally on $\varphi$.

The most general form that $a_m(\sqrt{\varphi})$ can take is $P + Q \sqrt{\varphi}$, where  $P$ and $Q$ are rational functions of $\varphi$ (although they can have algebraic dependence on the other kinematic variables). However, it is often more convenient to work with letters of the form $\smash{a_m(\sqrt{\varphi}) = \frac {P + Q \sqrt{\varphi}}{P - Q \sqrt{\varphi}}}$, which are odd under the transformation $\sqrt{\varphi} \to -\sqrt{\varphi}$. The differential of a letter taking this form is given by
\begin{align}
   d \log \frac {P + Q \sqrt{\varphi}}{P - Q \sqrt{\varphi}} & =
-\frac {2 Q \sqrt{\varphi}}{P^2 - Q^2 \varphi} d P
\\ & \hspace{-1cm} +\frac {2 P \sqrt{\varphi}}{P^2 - Q^2 \varphi} d Q
+ \frac {P Q}{P^2 - Q^2 \varphi} \frac {d \varphi}{\sqrt{\varphi}}\, . 
\nonumber
\end{align}
In the limit $\varphi \to 0$, the leading contribution is thus
\begin{equation}
d \log \frac {P + Q \sqrt{\varphi}}{P - Q \sqrt{\varphi}} = \frac {d \varphi}{\sqrt{\varphi}} \frac {Q}{P} + \mathcal{O}(\sqrt{\varphi}) \, .
\end{equation}
When pulled back to the path $\sigma$, the variables $P$, $Q$, and $\varphi$ all become functions of $t$.  Taking $\varphi(t) = (1 - t) + t \varphi$, we then have
\begin{equation}
    \sigma^*\Bigl(d \log \frac {P + Q \sqrt{\varphi}}{P - Q \sqrt{\varphi}}\Bigr) \sim
     - \frac {Q(t)}{P(t)} \frac {d t}{\sqrt{\varphi(t)}}.
\end{equation}
at leading order in $\varphi$.

We want to study the behavior of the integral~\eqref{eq:iterated_integral_different_form} in the $\varphi \to 0$ limit. Due to the way we have parameterized the integration contour, these singularities will arise near the $t \to 1$ integration boundary.  Assuming that $P(t)$ and $Q(t)$ are regular when $t$ approaches 1, we can replace them by $P(1)$ and $Q(1)$. Similarly, we need to consider $U(t)$ and $V(t)$ as $t \to 1$. $U(1)$ will generically be a non-zero function of the remaining kinematic degrees of freedom and can be pulled out of the integral over $t$, while $V(1)$ vanishes due to the fact that the integration contours in~\eqref{eq:V_def} all vanish. To find the leading non-analytic behavior of $\I_G(\varphi \to 0)$, we expand $V(t\to1)$ to leading order. As explained in~\cite{Hannesdottir:2021kpd}, the first non-zero contribution will occur at order $(t - 1)^{n-m}$, since there are $n-m$ vanishing integrals. Defining $q = n - m$ for convenience, we are left with the integral
\begin{equation}
    \int_0^1 \frac {d t}{\sqrt{(1 - t) + t \varphi}} (t - 1)^q.
\end{equation}
Changing the integration variable to $u = (1 - t) + t \varphi$, this integral can be performed by binomially expanding the integrand:
\begin{multline}
\int_{\varphi}^1 \frac {d u}{\sqrt{u}} (u - \varphi)^q =
 \int_{\varphi}^1 d u \sum_{k = 0}^q \binom{q}{k} u^{k - \frac 1 2} (-1)^{q - k} \varphi^{q - k} \\
 = \sum_{k = 0}^q \binom{q}{k} \frac 1 {k + \frac 1 2} (1 - \varphi^{k + \frac 1 2}) (-1)^{q - k}\varphi^{q - k} \\
 =  -\varphi^{q + \frac 1 2} \sum_{k = 0}^q \frac {(-1)^{q - k}}{k + \frac 1 2} \binom{q}{k} +
\sum_{k = 0}^q \frac {(-1)^{q - k}}{k + \frac 1 2} \binom{q}{k} \varphi^{q - k}.
\end{multline}
Note that second sum in the last line is over integer powers of $\varphi$, and is therefore independent of $\sqrt{\varphi}$.  Correspondingly, the leading non-analytic behavior is
\begin{align}
\int_0^1 \frac {d t}{\sqrt{(1 - t) + t \varphi}} (t - 1)^q & \sim
-\varphi^{q + \frac 1 2} \sum_{k = 0}^q \frac {(-1)^{q - k}}{k + \frac 1 2} \binom{q}{k} 
\nonumber
\\
& \hspace{-2cm} = 2 (-1)^{q+1} \varphi^{q + \frac 1 2} \frac {(2 q)!!}{(2 q + 1)!!} \, ,
\label{eq:symexp}
\end{align}
namely it goes as $\varphi^{q + \frac 1 2} = \varphi^{n-m + \frac 1 2}$.

In the case of Feynman integrals with generic masses, we can immediately compare this to the order at which algebraic branch cuts are predicted to appear, which can be read off of equation~\eqref{eq:Landau_expansion}. Comparing this result for non-integer $\gamma$ to equation~\eqref{eq:symexp}, we conclude that
the square root branch point associated with  $\varphi = 0$ and identified with a Landau diagram with Landau exponent $\gamma$ must appear at least $\gamma - \frac 1 2$ entries from the end of the symbol. For instance, when $\gamma=\frac{3}{2}$ we have that $\sqrt{\varphi}$ cannot appear in the last entry of the symbol. Note that when $\gamma + \frac 1 2$ is a negative integer, however, there is no constraint on where $\sqrt{\varphi}$ can appear.

In the above analysis, we have placed no restriction on the dependence of $U(t)$, but have assumed that $V(t)$ is free of algebraic dependence on $\varphi$. We have also assumed that $U(t)$ and $V(t)$ evaluate to nonzero constants as $t \to 1$. In many Feynman integrals, however, there exist letters that approach 1 as $\varphi \to 0$, such as $a_i=1-\varphi$. The appearance of these letters can suppress the values of $U(t)$ and $V(t)$ in the $\varphi \to 0$ limit, changing the order at which the algebraic branch cut at $\varphi\to 0$ first appears. In practice, then, when working with examples involving pairs of letters that alternately vanish and give rise to branch points at the same kinematic location, it often proves easiest to directly expand one's ansatz for the symbol in the $\varphi \to 0$ limit. One thereby derives a prediction for the leading order at which the branch points in each symbol term will arise, which can be compared to the expectation one has from approximating the Feynman integral itself in the same limit.

\section{Asymptotic Behavior of Feynman Integrals with Numerators}
\label{appendix:asymptotics_with_numerators}

In this appendix we derive a generalized formula for the Landau exponent of singular limits in which the numerator of a Feynman integral vanishes. That is, we consider generic integrals of the form
\begin{equation}
    \label{eq:pham_integral}
    \I (p) = \int_{h_k} \,\, \omega \,\, \,\,
    \frac{\prod_{j = 1}^\mu (-s_j(k, p))^{-\nu_j}}{\prod_{i = \mu + 1}^m s_i(k, p)^{\nu_i}},
\end{equation}
where $s_1, \ldots, s_m$ are surfaces that take part in a simple pinch, and $\omega$ collects the integration measure over the loop momenta $k$ and any additional factors that do not participate in the pinch. Our goal will be to find the expansion of this integral near the kinematic hypersurface $\varphi=0$ on which this pinch occurs. We thus allow the contour over the loop momenta $h_k$ to be arbitrary, beyond requiring that it contains the integration region in which the pinch occurs. We also assume that $\nu_j<0$ for $j \in \{1, \ldots, \mu\}$ so that these $s_j$ can be thought of as numerators. Intuitively, the presence of these numerator factors that vanish at the location of the pinch are expected to dampen the singular behavior of $\I(p)$ on $\varphi = 0$. 

We start by combining the factors $s_1, \ldots, s_m$ into a single denominator factor. To do so, we need separate formulas for the factors that appear in the numerator and denominator. For the denominators, we can use
\begin{equation}
    \frac{1}{x^\nu} = \frac{1}{\Gamma(\nu)} \int_0^\infty \d \alpha \,  \alpha^{\nu - 1} e^{-x \alpha},
    \label{eq:exp1}
\end{equation}
which is a consequence of a simple change of variables.  This formula is valid for $\Re x > 0$ and $\Re \nu > 0$, since otherwise we have a divergence either at $\alpha \to \infty$ or $\alpha \to 0$. For the numerators, on the other hand, we find the following formula useful:
\begin{equation}
  \int_{\lvert \alpha \rvert = 1} e^\alpha \alpha^{-\nu - 1} \d \alpha = \frac {2 \pi i}{\nu!},
\end{equation}
which can be shown by the residue theorem.  It implies that 
\begin{equation}
  (-x)^{-\nu} = \frac{(-\nu)!}{2 \pi i} \int_{\lvert \alpha \rvert = 1} \d \alpha \, \alpha^{\nu - 1}  e^{-\alpha x}  ,
  \label{eq:exp2}
\end{equation}
for $\nu \in \mathbb{Z}_{\leq 0}$.

Making use of these formulas, we find
\begin{align}
    & \frac{\prod_{i = 1}^\mu (-s_i)^{-\nu_i}}{\prod_{i = \mu + 1}^m s_i^{\nu_i}} = \frac{\prod_{i = 1}^\mu (-\nu_i)!}{(2 \pi i)^\mu \prod_{i = \mu + 1}^m \Gamma(\nu_i)} \\ & \hspace{2.2cm} \times \int \prod_i \alpha_i^{\nu_i - 1} e^{-\sum_i s_i \alpha_i} \prod_i \d \alpha_i.
    \nonumber
\end{align}
Performing an integration over one of the $\alpha_i$ variables, this becomes
\begin{multline}
 \frac{\prod_{i = 1}^\mu (-s_i)^{-\nu_i}}{\prod_{i = \mu + 1}^m s_i^{\nu_i}} = 
     \frac{\Gamma(\sum_{i = 1}^m \nu_i) \prod_{i = \mu + 1}^m s_i^{\nu_i}}{(2 \pi i)^\mu \prod_{i = \mu + 1}^m \Gamma(\nu_i)} 
     \\ \times \int_{h_\alpha}  \!\!\!\! \frac{\delta( \alpha_m - 1)\prod_i \alpha_i^{\nu_i - 1} \d \alpha_i}{\bigl(\sum_{i=1}^m \alpha_i s_i\bigr)^{\sum_i \nu_i}},
\end{multline}
where $h_\alpha$ is the contour specified by~\eqref{eq:exp1} and~\eqref{eq:exp2}, namely a product of $\mu$ circles with $m - \mu$ copies of $\mathbb{R}_+$.

Plugging this result into equation~\eqref{eq:pham_integral} we find
\begin{multline}
    \I (p) = \frac {\Gamma(\nu) \prod_{i=1}^\mu (-\nu_i)!}{(2 \pi i)^\mu \prod_{i = \mu + 1}^m \Gamma(\nu_i)} \\ \times  \int_{h'_\alpha \times h_k} \!\!\!
     \frac{\d^{m - 1} \alpha \, \omega \prod_{j=1}^m \alpha_j^{\nu_j-1}}{\left[ \sum_{i = 1}^m \alpha_i s_i\right]^\nu},\label{eq:I_feyn_param_with_num}
\end{multline}
where $h'_\alpha$ is obtained from $h_\alpha$ by setting $\alpha_m = 1$ and we have defined $\nu = \sum_{i = 1}^m \nu_i$. To simplify our notation, we define the function $M(k,p,\alpha) = \sum_{j = 1}^m \alpha_j s_j(k,p)$, where we leave implicit the fact that we have already fixed one of the $\alpha$ variables using the delta function (which we have taken to be $\alpha_m$ without loss of generality). 

The nature of the singularity near $\varphi = 0$ will be determined by the behavior of $\I(p)$ near the critical point $(k^\ast, p^\ast, \alpha^\ast)$ where the pinch occurs. Correspondingly, we expand the denominator around this point. Since, by assumption, all of the factors $s_j$ in $M$ take part in the pinch, all of the terms in this expansion that are linear with respect to one of the integration variables will vanish. The behavior of the integrand near the critical point can thus be captured by expanding $M$ to quadratic order:\footnote{Here we assume that the singularity can be reached with all of the differences $\delta \alpha_i$ and $\delta k_j$ vanishing at the same rate. This assumption is further discussed and justified in~\cite{Landau:1959fi,PolkinghorneScreaton,pham,BSMF_1959__87__81_0}.}
\begin{multline}
 M(k, p, \alpha) =
 M^\ast + 
  \sum_{i = 1}^m \alpha_{i}^\ast \Big( \frac {\partial s_i}{\partial p_a} \Big)^\ast \delta p_a \\
  + \frac 1 2 \sum_{i = 1}^m \alpha_{i}^\ast \Big( \frac {\partial^2 s_i}{\partial p_a \partial p_{a'}} \Big)^\ast \delta p_a \delta p_{a'} 
   + \sum_{i = 1}^{m - 1} \left( \frac {\partial s_i}{\partial k_j} \right)^\ast \delta \alpha_i \delta k_j \\
  + \frac 1 2 \sum_{i = 1}^m \alpha_{i}^\ast \Big( \frac {\partial^2 s_i}{\partial k_j \partial k_{j'}} \Big)^\ast \delta k_j \delta k_{j'}
  + \cdots \label{eq:M_expansion}
\end{multline}
Here, we have labeled the quantities that are evaluated at the critical point with superscript $\ast$, and have defined $\delta q = q - q^\ast$ for $q \in \{k,p,\alpha\}$. Repeated momentum indices should be understood to indicate a sum over all momenta of the indicated type, as well as component-wise contraction via the spacetime metric.  %

We can rewrite the expansion in~\eqref{eq:M_expansion} more schematically as
\begin{equation} M = \ell(p) + \frac 1 2 M_{A B}(\xi^\ast) \xi_A \xi_B + \cdots, \label{eq:M_expansion_xi}
\end{equation}
where we have defined 
\begin{align}
  \ell(p) & = \sum_{i = 1}^m \alpha_{i}^\ast \Big( \frac {\partial s_i}{\partial p_a} \Big)^\ast \delta p_a  \\ & +
    \frac 1 2 \sum_{i = 1}^m \alpha_{i}^\ast \Big( \frac {\partial^2 s_i}{\partial p_a \partial p_{a'}} \Big)^\ast  \delta p_a \delta p_{a'}+ \cdots
    \nonumber \, ,
\end{align}
and used the fact that $M^\ast = 0$ (since $s_i^\ast = 0$ for all $i$). The indices $A$ and $B$ run from $1$ to $n + m - 1$, while $\xi = (\alpha_1, \dotsc, \alpha_{m - 1}, k_1, \dotsc, k_n)$ is a vector that collects together all integration variables (so $n = D L$ if we are expanding around all the integration variables corresponding to $L$ loop momenta in $D$ dimensions). The symmetric matrix 
\begin{equation}
 M_{AB} =
  \begin{pmatrix} 0 & \frac {\partial s_A}{\partial \xi_B} \\
    \frac {\partial s_B}{\partial \xi_A} & \ \ \sum_{i = 1}^m \alpha_{i}^\ast \frac {\partial^2 s_i}{\partial \xi_A \partial \xi_B}
  \end{pmatrix}
\end{equation}
is just the Hessian, in which we have removed the row and column corresponding to $\alpha_m$. 

Evaluating $\I(p)$ to leading order near the critical point thus  reduces to evaluation of an integral of the form
\begin{equation}
  \label{eq:quadratic_int}
  \int_{h_\xi} \frac {\d^p \xi}{(c^2 + \xi^T D \xi)^q}\, .
\end{equation}
When $D$ is a positive-definite $p \times p$ symmetric matrix, the contour $h_\xi$ will be defined by $\xi^t D \xi \leq R^2$, for some small constant $R$. Moreover, since $D$ is symmetric, it can be diagonalized to define $\smash{D = U^T \operatorname{diag}(\sigma_1, \dotsc, \sigma_p) U = V^T V}$, where $\sigma_1,\dots,\sigma_p$ are the eigenvalues of $D$, $U$ is an orthogonal matrix and $\smash{V = \operatorname{diag}(\sqrt{\sigma_1}, \dotsc, \sqrt{\sigma_p}) U}$.  Then $\smash{\det D = (\det V)^2}$ so $\smash{\sqrt{\det D} = \det V}$.  We can thus make a change of coordinates $\xi = V^{-1} y$, upon which the domain of integration $h_y$ will be given by the equation $\sum_{i = 1}^p y_i^2 \leq R^2$.  We then have
\begin{equation}
  \int_{h_\xi} \frac {\d^p \xi}{(c^2 + \xi^T D \xi)^q} =
  \frac 1 {\sqrt{\det D}} \int_{h_y} \frac {\d^p y}{(c^2 + \sum_{i = 1}^n y^2)^q}.
\end{equation}
Going to spherical coordinates, we obtain
\begin{equation}
  \int_{h_\xi} \frac {\d^p \xi}{(c^2 + \xi^T D \xi)^q} = \frac {S_{p - 1}}{\sqrt{\det D}} \int_0^\rho \frac {r^{p - 1} \d r}{(c^2 + r^2)^q},
\end{equation}
where $S_{p - 1}$ is the area of the $(p{-}1)$-dimensional sphere. 
As long as $2 q - p > 0$ the integral is convergent as $c \to 0$.  In this limit,
we make a change of variables $r \to c\tau$ and obtain
\begin{equation}
  \frac {S_{p - 1} c^{\, p - 2 q}}{\sqrt{\det D}} \int_0^\infty \frac {\tau^{p - 1} \d \tau}{(1 + \tau^2)^q} =
  \frac {c^{\, p - 2 q} \pi^{\frac p 2} \Gamma(q - \frac p 2)}{\Gamma(q) \sqrt{\det D}},
\end{equation}
where we have used that the area of the sphere can be written as $\smash{S_{p-1} =  {2 \pi^{\frac {p} 2}}/{\Gamma(\frac {p} 2)}}$.
Altogether, then, we have that the leading non-analytic contribution to the integral in~\eqref{eq:quadratic_int} goes as 
\begin{equation}
  \int \frac {\d^p \xi}{(c^2 + \xi^T D \xi)^q} \sim \frac {c^{\, p - 2 q} \pi^{\frac p 2} \Gamma(q - \frac p 2)}{\Gamma(q) \sqrt{\det D}} \, ,
\end{equation}
in the $c \to 0$ limit, as long as $D$ as symmetric and positive-definite.

Applying this result to~\eqref{eq:I_feyn_param_with_num}, we find that the leading non-analytic contribution to $\I(p)$ as $\ell(p) \to 0$ is given by 
\begin{multline} \label{eq:landau_exponent_final}
  \I (p) \sim  \frac {\prod_{i = 1}^m {\alpha_{i}^\ast}^{\nu_i - 1}}{(2 \pi i)^\mu \prod_{i = \mu + 1}^m  \Gamma(\alpha_i)}  \\  \times \pi^{\frac {n + m - 1} 2} \Gamma(\alpha - \tfrac{n + m - 1} {2} )  \ell(p)^{\frac {n + m - 1} 2 - \sum_{i=1}^m \nu_i}.
\end{multline}
where we have replaced the $\alpha_i$ variables in the numerator by the value they take at the critical point, which is valid at leading order. %

The key piece of information that in~\eqref{eq:landau_exponent_final} is the exponent of $\ell(p)$, which tells us how $\I(p)$ behaves near the singularity at $\varphi = \ell(p) \to 0$. Namely, the Landau exponent for this singularity is given by
\begin{equation}
\label{eq:landau_exp2}
    \gamma = \frac{n+m-1}{2} - \sum_{i=1}^m \nu_i \,.
\end{equation}
where we recall that $n$ is the number of integrals coming from the loop momenta (which is usually $D L$), $m$ is the original number of numerator and denominator factors in~\eqref{eq:pham_integral}, and $\nu_i$ are the (inverse) powers to which these factors were raised. A different proof of this formula was given by Pham in~\cite{pham1968singularities} and in section 2.1 of~\cite{pham2011singularities}, but our proof has the benefit of being conceptually simpler.

\section{Absence of mixed second-type singularity in \texorpdfstring{$\D=4$}{D=4}}
\label{app:second-type}
In this appendix, we identify the singularity at $s t^2-2 s t m^2+s m^4-4 t^2 m^2=0$ as a mixed second-type singularity, meaning that it comes from an integration region in which one of the loop momenta goes to infinity while the other one stays finite. Moreover, we show that this singularity does not appear in strictly four dimensions, consistent with what was observed in the expansion of this integral around four dimensions in~\cite{Caron-Huot:2014lda}.

We begin by considering the double box in dual coordinates~\cite{Broadhurst:1993ib,Drummond:2007aua}, using the labeling
\begin{equation}
    \begin{gathered}
    \raisebox{0pt}[\height][\depth]{\hspace{-30pt}%
    \begin{tikzpicture}[scale=0.6, thick]
    \node at (-3,0) {$x_2$};
    \node at (5,0) {$x_4$};
    \node at (1,-2.5) {$x_1$};
    \node at (1,2.5) {$x_3$};
    \node at (0,0) {$x_A$};
    \node at (2,0) {$x_B$};
    \draw[line width=1.2] (-1,1) -- (3,1) -- (3,-1) -- (-1,-1) --cycle;
    \node[scale=0.7] at (0,1.55) {$\alpha_2$};
    \node[scale=0.7] at (2,1.55) {$\alpha_4$};
    \node[scale=0.7] at (0,-1.55) {$\alpha_1$};
    \node[scale=0.7] at (2,-1.55) {$\alpha_3$};
    \node[scale=0.7] at (-1.5,-0.5) {$\alpha_5$};
    \node[scale=0.7] at (1.4,-0.5) {$\alpha_6$};
    \node[scale=0.7] at (3.5,-0.5) {$\alpha_7$};
    \draw[dashed] (1,1) -- (1,-1);
    \draw[dashed] (-1,1) -- ++(135:1.8);
    \draw[dashed] (3,1) -- ++(45:1.8);
    \draw[dashed] (3,-1) -- ++(-45:1.8);
    \draw[dashed] (-1,-1) -- ++(-135:1.8);
    \end{tikzpicture}
    }
    \end{gathered}
\label{fig:dbox_labels}
\end{equation}
In these variables, the integral is given by
\begin{align}
\I (p_i) & = \int \frac {\d^\D x_A \d^\D x_B}{(x_{1A}^2{-}m^2)(x_{2A}^2{-}m^2)(x_{3A}^2{-}m^2)}
\nonumber
\\
& \times \frac{1}{(x_{3B}^2{-}m^2)(x_{4B}^2{-}m^2)(x_{1B}^2{-}m^2)x_{AB}^2} \,,
\end{align}
where we have used the standard shorthand notation $x_{ij}^2 \equiv (x_i-x_j)^2$. We now bring the point at infinity in the $x_B$ variable to zero through the change of variables
\begin{equation}
  y_B = \frac{x_B}{x_B^2} \,.
\end{equation}
This corresponds to making the replacements $\d^\D x_B \to \frac{\d^\D y_B}{(y_B^2)^\D}$, and $x_{Bi}^2 \to \frac{1}{y_B^2} - 2 \frac{y_B \cdot x_i}{y_B^2}+x_i^2$. Pulling a factor of $1/y_B^2$ out of the last four propagators, $\I (p_i)$ becomes
\begin{multline}
\I (p_i) = \int \frac {\d^\D x_A \d^\D y_B}{(y_B^2)^{\D-4}} \frac{1}{(x_{1A}^2{-}m^2)(x_{2A}^2{-}m^2)(x_{3A}^2{-}m^2)} \\ \times 
\frac{1}{[1 {-} 2 x_3 \cdot y_B {+} (x_3^2 {-} m^2) y_B^2][1 {-} 2 x_4 \cdot y_B {+} (x_4^2 {-} m^2) y_B^2]} \\ \times
\frac{1}{[1 {-} 2 x_1 \cdot y_B {+} (x_1^2 {-} m^2) y_B^2](1 {-} 2 x_A \cdot y_B {+} x_A^2 y_B^2)} \, ,
\label{eq:dualint_secondtype}
\end{multline}
where we highlight that the factor of $1/{(y_B^2)^{\D-4}}$ drops out in four dimensions.

As usual, to identify where this integral can become singular, we  scan over all possible tuples of denominator factors that can vanish and thereby pinch the integration contour. We find that the solution to the Landau equations that identifies the singularity at $s t^2-2 s t m^2+s m^4-4 t^2 m^2=0$ comes from the following Landau diagram:
\begin{equation}
    \begin{gathered}
    \raisebox{0pt}[\height][\depth]{\hspace{-30pt}%
    \begin{tikzpicture}[scale=0.5, thick]
    \draw [decorate,decoration={brace,amplitude=5pt,mirror}]
    (-2.5,1.8) -- (-2.5,-1.8);
    \node at (-4.5,0) {$s$};
    \draw [decorate,decoration={brace,amplitude=5pt,mirror}]
    (-1.8,-2.5) -- (1.8,-2.5);
    \node at (0,-4) {$t$};
    \draw[line width=1.2] (-1,1) -- (1,1);
    \draw[line width=1.2] (1,-1) -- (-1,-1) -- (-1,1);
    \node[] at (0,1.8) {$m$};
    \node[white] at (0,2.5) {$.$};
    \draw[dashed] (1,1) to[in=130,out=-130] (1,-1);
    \draw[line width=1.2] (1,1) to[in=50,out=-50] (1,-1);
    \draw[dashed] (-1,1) -- ++(135:1.5);
    \draw[dashed] (1,1) -- ++(45:1.5);
    \draw[dashed] (1,-1) -- ++(-45:1.5);
    \draw[dashed] (-1,-1) -- ++(-135:1.5);
    \end{tikzpicture}
    }
    \end{gathered}
\end{equation}
meaning that the denominators $(x_{1B}^2-m^2)$ and $(x_{3B}^2-m^2)$ do not take part in the pinch. Concretely, the on-shell conditions for this Landau diagram read
\begin{gather}
    x_{A2}^2 = x_{A4}^2 = m^2, \qquad
    y_B^2 = 0, \\
    1 - 2 x_A \cdot y_B + x_A^2 y_B^2 = 0, \\
    1 - 2 x_4 \cdot y_B + (x_4^2 - m^2) y_B^2 = 0.
\end{gather}
The Landau loop equations read
\begin{align}
    & \alpha_5 x_{A2} + \alpha_1 x_{A4} + \alpha_3 x_{A1} = \alpha_6 y_B, \\
    & \alpha_0 y_B + \alpha_6 (y_B x_A^2 - x_A) + \alpha_7 (y_B (x_4^2 - m^2) - x_4) = 0,
\end{align}
where the labelings of $\alpha_i$ are as in~\eqref{fig:dbox_labels}, and we have assigned $y_B$ a parameter $\alpha_0$.

To solve these equations, we start by dotting $y_B$ into the second equation. %
Using the on-shell conditions, this equation can be reduced to $\alpha_6 + \alpha_7 = 0$.  The second equation becomes
\begin{equation}
    y_B (\alpha_0 + \alpha_6 x_A^2 + \alpha_7 (x_4^2 - m^2)) = \alpha_6 x_{A4}.
\end{equation}
This implies that $x_{A4}^2 = 0$, since $y_B^2 = 0$. Solving for $y_B$ and plugging into the first Landau equation we find that $\langle A 1 2 3 4\rangle = 0$ since these vectors are linearly dependent.  Here we have used the notation $\langle v_i v_j v_k v_l v_m\rangle$ for the oriented volume of the four-simplex with vertices of coordinates $v_i, v_j, v_k, v_l, v_m$.  It can be obtained as the $5 \times 5$ determinant of the five-vectors obtained by appending a component $1$ at the beginning of the four-vectors $v_i, \dotsc, v_m$.

Plugging in the fact that $x_{A3}^2 = 0$ in the expression for $\langle A 1 2 3 4\rangle$, we find
\begin{equation}
    \frac{1}{16} \, {\left(s t^2-2 s t m^2+s m^4-4 t^2 m^2\right)} s = 0 \, ,
\end{equation}
thereby identifying the singularity locus we were after. We highlight, however, that we needed to use the vanishing condition $y_B^2 = 0$ to derive this solution. As this denominator factor does not exist in strictly four dimensions (due to dual conformal invariance), this singularity also does not exist in four dimensions.

\section{Constructing Algebraic Symbol Letters}%
\label{sec:symbols-from-sqrt}
In this appendix, we illustrate how one can systematically search for algebraic symbol letters that are consistent with one's expectations from Landau analysis. We focus on the example of the bubble integral in two dimensions, whose predicted singular behavior is cataloged in Table~\ref{tab:bootstrap_bubble}. 
Given that the only two expected square root branch points in this integral arise at the threshold and pseudothreshold, we are led to consider a symbol-letter ansatz of the form
\begin{multline}
L= P_1(s, m_1, m_2) \sqrt{s - (m_1 + m_2)^2} \\  + P_2(s, m_1, m_2) \sqrt{s + (m_1 + m_2)^2},
\end{multline}
where $P_1$ and $P_2$ are polynomials. What we will show in this Appendix is that for the logarithmic singularities ($L=0$) to be at the allowed locations ($m_1^2=0$ or $m_2^2=0$ from Table~\ref{tab:bootstrap_bubble}), the only independent letters are
\begin{equation}
     L = \sqrt{s - (m_1 + m_2)^2}  \pm \sqrt{s + (m_1 + m_2)^2}
\end{equation}
That is, the only possibility for the polynomials $P_1$ and $P_2$ is that they are $\pm 1$. 

To begin, we note that
\begin{align}
  & \Bigl(P_1(s, m_1, m_2) \sqrt{s - (m_1 + m_2)^2}
  \nonumber
  \\ & \qquad + P_2(s, m_1, m_2) \sqrt{s + (m_1 + m_2)^2}\Bigr) 
  \nonumber
  \\
  & \times \Bigl(P_1(s, m_1, m_2) \sqrt{s - (m_1 + m_2)^2}
  \label{eq:product-bubble-ansatz}
  \\
  & \qquad
  - P_2(s, m_1, m_2) \sqrt{s + (m_1 + m_2)^2}\Bigr)
  \nonumber
  \\
  & = P_1^2 (s - (m_1 + m_2)^2) - P_2^2 (s - (m_1 - m_2)^2).
  \nonumber
\end{align}
If this combination is a nontrivial polynomial in $s$,  we will encounter logarithmic singularities at the roots of that polynomial. Since no such singularities are expected, we should obtain a result that does not depend on $s$. Conversely, the right hand side of equation~\eqref{eq:product-bubble-ansatz} may depend on $m_1$ and $m_2$, since we do expect logarithmic singularities when $m_1 = 0$ or $m_2 = 0$ (and also when $m_1 \to \infty$ and $m_2 \to \infty$).  Since there are no other singularities at special values of the masses, we deduce that this polynomial should be a product of powers of $m_1$ and $m_2$. By symmetry, these powers should be equal.

Looking at the expression in the right hand side of equation~\eqref{eq:product-bubble-ansatz}, we are led to ask the following mathematical question.  Given $\alpha, \beta, \gamma \in \mathbb{C}$ (or in some polynomial ring such as $\mathbb{C}[m_1, m_2]$), when can we find two polynomials $p, q$ such that
\begin{equation}
  p^2 (x - \alpha) - q^2 (x - \beta) = \gamma .
\end{equation}
If $p, q$ are of degree zero, we have the obvious solution $p = q = 1$ and $\gamma = \beta - \alpha$. If $p = x + p_0$ and $q = x + q_0$, we have
\begin{gather}
  p_0 = -\frac {\alpha + 3 \beta} 4, \qquad
  q_0 = -\frac {3 \alpha + \beta} 4, \qquad
  \gamma = -\frac 1 {16} (\alpha - \beta)^2.
\end{gather}
If $p = x^2 + p_1 x + p_0$ and $q = x^2 + q_1 x + q_0$, then we find
\begin{gather}
  p_1 = -\frac 3 4 \alpha - \frac 5 4 \beta, \qquad
  p_0 = \frac 1 {16} \alpha^2 + \frac 5 8 \alpha \beta + \frac 5 {16} \beta^2, \\
  q_1 = -\frac 5 4 \alpha - \frac 3 4 \beta, \qquad
  q_0 = \frac 5 {16} \alpha^2 + \frac 5 8 \alpha \beta + \frac 1 {16} \beta^2, \\
  \gamma = -\frac 1 {256} (\alpha - \beta)^5.
\end{gather}
In all these cases we have a unique solution.  We will provide an interpretation of this solution below.

Consider now
\begin{equation}
  (\sqrt{x - \alpha} \pm \sqrt{x - \beta})^{2 n + 1} = p_n \sqrt{x - \alpha} \pm q_n \sqrt{x - \beta},
\end{equation}
where
\begin{multline}
  (\sqrt{x - \alpha} \pm \sqrt{x - \beta})^{2 n + 1} = \\
  \sum_{k = 0}^{2 n + 1} \binom{2 n + 1}{k} (x - \alpha)^{\frac k 2} (\pm 1)^{k - 1} (x - \beta)^{\frac {2 n + 1 - k} 2} = \\
  \pm \sqrt{x - \alpha} \sum_{k = 0}^n \binom{2 n + 1}{2 k} (x - \alpha)^k (x - \beta)^{n - k} + \\
  \sqrt{x - \beta} \sum_{k = 0}^n \binom{2 n + 1}{2 k + 1} (x - \alpha)^k (x - \beta)^{n - k}.
\end{multline}

We therefore have
\begin{gather}
  p_n(x) = \pm \sum_{k = 0}^n \binom{2 n + 1}{2 k} (x - \alpha)^k (x - \beta)^{n - k}, \\
  q_n(x) = \sum_{k = 0}^n \binom{2 n + 1}{2 k + 1} (x - \alpha)^k (x - \beta)^{n - k}.
\end{gather}

With this definition we indeed have
\begin{equation}
  p_n^2 (x - \alpha) - q_n^2 (x - \beta) = (\beta - \alpha)^{2 n + 1}.
\end{equation}

We thus find
\begin{gather}
  p_1(x) = 4 x - 3 \alpha - \beta, \\
  q_1(x) = 4 x - \alpha - 3 \beta, \\
  p_2(x) = 16 x^2 - 4 (5 \alpha + 3 \beta) x + (5 \alpha^2 + 10 \alpha \beta + \beta^2), \\
  q_2(x) = 16 x^2 - 4 (3 \alpha + 5 \beta) x + (\alpha^2 + 10 \alpha \beta + 5 \beta^2).
\end{gather}
If we normalize the $p_n$ such that their leading coefficient is 1, we obtain the same solution as before.

We now proceed to show in general that a solution to the constraints formulated above with non-trivial polynomials $p$ and $q$ is a power of the solution obtained for $p = q = 1$.  Therefore, such candidate symbol letters are not independent and can be discarded.

We start with a key Lemma:
\begin{lemma}%
  \label{lem:invertible}
  Given $f$ a Laurent polynomial $f \in \mathbb{C}[u, u^{-1}]$, if $f$ is invertible, then there exists $k \in \mathbb{Z}$ such that $f(u) = f_k u^k$.
\end{lemma}
\begin{proof}
  Let us take
\begin{equation}
    f = \sum_k f_k u^k = f_{k_L} u^{k_L} + \cdots + f_{k_U} u^{k_U},
\end{equation}
  with $k_U \geq k_L$.  Since $f$ is invertible, there exists $g \in \mathbb{C}[u, u^{-1}]$ such that $f g = 1$.  If
\begin{equation}
    g = \sum_l g_l u^l = g_{l_L} u^{l_L} + \cdots g_{l_U} u^{l_U},
\end{equation}
  with $l_U \geq l_L$, then we have that the maximum degree of $f g$ is $k_U + l_U$ while the minimum degree is $k_L + l_L$.  Since $f g = 1$ we must have $k_U + l_U = 0$ and $k_L + l_L = 0$.  In particular, we have $k_U + l_U = k_L + l_L$.

  From the conditions $k_U \geq k_L$ and $l_U \geq l_L$ on the degrees above, we have $k_U + l_U \geq k_L + l_L$.  Since, as we have shown, $k_U + l_U = k_L + l_L$, we must have $k_U = k_L$ and $l_U = l_L$.  This means that the sums defining $f$ and $g$ reduce to a single term.
\end{proof}

\begin{lemma}%
  \label{lem:solution}
  If $p, q \in \mathbb{C}[u, u^{-1}]$ and
  \begin{equation}
    p(u) (u + u^{-1}) + q(u) (u - u^{-1}) = u^{-(2 n + 1)},
  \end{equation}
  such that $p(u) = p(u^{-1})$ and $q(u) = q(u^{-1})$, then
  \begin{align}
    p(u) & = \frac {u^{2 n + 1} + u^{-(2 n + 1)}}{2 (u + u^{-1})}
    \nonumber
    \\ 
    & = \frac 1 2 (u^{2 n} - u^{2 n - 2} + \cdots + u^{-2 n}), \\
    \nonumber
    q(u) & = -\frac {u^{2 n + 1} - u^{-(2 n + 1)}}{2 (u - u^{-1})}
    \\ 
    & = -\frac 1 2 (u^{2 n} + u^{2 n - 2} + \cdots + u^{-2 n}).
  \end{align}
\end{lemma}
\begin{proof}
  Putting $u \to u^{-1}$ in the formula in the statement we find
  \begin{equation}
    p(u^{-1}) (u^{-1} + u) + q(u^{-1}) (u^{-1} - u) = u^{2 n + 1}.
  \end{equation}
  Using $p(u^{-1}) = p(u)$ and $q(u^{-1}) = q(u)$ and adding to the equality above, we find
  \begin{equation}
    p(u) = \frac {u^{2 n + 1} + u^{-(2 n + 1)}}{2 (u + u^{-1})}.
  \end{equation}
  Subtracting instead of adding, we find
  \begin{equation}
    q(u) = \frac {-u^{2 n + 1} + u^{-(2 n + 1)}}{2 (u - u^{-1})}.
  \end{equation}
\end{proof}

\begin{theorem}
  Given $P, Q \in \mathbb{C}[x]$ and $\alpha, \beta, \gamma \in \mathbb{C}$, if
  \begin{equation}
    P^2(x) (x - \alpha) - Q^2(x) (x - \beta) = \gamma,
  \end{equation}
  then there exists $n \in \mathbb{Z}$ such that, up to a rescaling of $P$ and $Q$ by a constant
  \begin{equation}
    P(x) \sqrt{x - \alpha} \pm Q(x) \sqrt{x - \beta} = (\sqrt{x - \alpha} \pm \sqrt{x - \beta})^{2 n + 1}.
  \end{equation}
\end{theorem}
\begin{proof}
  The equality $P^2(x) (x - \alpha) - Q^2(x) (x - \beta) = \gamma$ can be written in a factorized form
  \begin{align}
    &\bigl(P(x) \sqrt{x - \alpha} + Q(x) \sqrt{x - \beta}\bigr) \\ &\quad \times
    \bigl(P(x) \sqrt{x - \alpha} - Q(x) \sqrt{x - \beta}\bigr)  = \gamma \, . \nonumber
  \end{align}
  We make the substitution
  \begin{equation}
    x = \frac {\alpha + \beta} 2 + \frac {\beta - \alpha} 4 (u^2 + u^{-2}).
  \end{equation}
  Using this we can rationalize the square roots
  \begin{gather}
    \sqrt{x - \alpha} = \frac {\sqrt{\beta - \alpha}} 2 (u + u^{-1}), \\
    \sqrt{x - \beta} = \frac {\sqrt{\beta - \alpha}} 2 (u - u^{-1}).
  \end{gather}
  We define
  \begin{equation}
    x(u) = \frac {\alpha + \beta} 2 + \frac {\beta - \alpha} 4 (u^2 + u^{-2}).
  \end{equation}
  If we denote $P(u) = P(x(u))$ and $Q(u) = Q(x(u))$, then the factorized form can be written
  \begin{multline}
    \big(P(u) (u + u^{-1}) + Q(u) (u - u^{-1})\bigr) \times \\
    \bigl(P(u) (u + u^{-1}) - Q(u) (u - u^{-1})\bigr) = \frac {4 \gamma}{\beta - \alpha}.
  \end{multline}
  Here $P(u)$ and $Q(u)$ are Laurent polynomials and since $x(u) = x(u^{-1})$ we have that $P(u) = P(u^{-1})$ and $Q(u) = Q(u^{-1})$.

  We can rescale $P$ and $Q$ such that the right hand side becomes unity.  Then, we can apply Lemma~\ref{lem:invertible} to conclude that there exists an integer $n$ such that
  \begin{equation}
    P(u) (u + u^{-1}) + Q(u) (u - u^{-1}) = \delta u^{2 n + 1}.
  \end{equation}
  The degree has to be odd since the degrees of monomials in $P$ and $Q$ are all even since $x$ depends only on even powers of $u$. Then, the conclusion follows by observing that
  \begin{gather}
    \sqrt{x - \alpha} + \sqrt{x - \beta} = u \sqrt{\beta - \alpha}, \\
    \sqrt{x - \alpha} - \sqrt{x - \beta} = u^{-1} \sqrt{\beta - \alpha}.
  \end{gather}
\end{proof}

Let us finish by highlighting that, although we started by focusing on the square roots that appear in the two-dimensional bubble integral, this proof holds in any situation in which a pair of roots $\sqrt{Q_1}$ and $\sqrt{Q_2}$ are expected to appear in a set of symbol letters. Namely, only a single independent symbol letter can be constructed that involves both of these roots. We expect that a similar proof strategy could also be used to establish a maximum number of independent letters that can be constructed out of a larger numbers of roots; however, we leave this question to future work.

\section{Maximal Cuts and Leading Singularities}\label{sec:maxcuts}

To each Feynman integral, we can associate a set of \emph{maximal cuts}. These maximal cuts are given by the absorption integrals that place a maximal number of propagators in the original diagram on shell (without over-constraining the loop momenta), by replacing the propagators by delta functions and fixing the direction of the energy flowing through the edge.
In favorable cases, these integrals can be computed in terms of elementary functions.  

The simplest cases arise when the number of integrations is equal to the number of cut propagators; then, the integral computation reduces to evaluating a Jacobian.
Feynman integrals may also have more propagators than there are integrations. In these cases, not all of propagators can be put on-shell, and there will be several maximal cuts.  
Finally, if there are more integrations than propagators, the maximal cut will itself be a non-trivial integral. In such cases, we can ask whether there exists a further sequence of residue contours that will allow us to evaluate the remaining integrals and get a nonzero answer. We call the expressions that remain after a maximum number of (nonzero) residues have been computed \emph{leading singularities}. In cases where all integrations can be evaluated in this way, the leading singularities will be algebraic functions of the external kinematics. Conversely, when obstructions occur to evaluating every integral as a residue, we learn about the types of special functions that may appear in the evaluation of the Feynman integral that go beyond $\d \log$ forms~\cite{Bourjaily:2022bwx}.

Let us analyze a simple situation one may encounter when computing leading singularities.  Assume that on-shell the calculation of the absorption integral reduces to computing
\begin{equation}
  \int_\gamma \frac {g(x) d x}{f(x)},
\end{equation}
where $f$ and $g$ are polynomials, and where the integration contour can be deformed to a linear combination (with integer coefficients) of contours going around the zeros of $f$.  If $\deg g > \deg f$ we compute the quotient $q$ and the remainder $r$ such that $g = q f + r$.  The integral becomes
\begin{equation}
  \int_\gamma (q + \frac r f) d x = \int_\gamma \frac {r d x} f,
\end{equation}
where we have used the fact that the integral of a polynomial along a closed contour vanishes.  If $\deg r = -1 + \deg f$, then we can write $r = \tilde{q} f' + \tilde{r}$, where $\tilde{q}$ is a number and $\tilde{r}$ is a polynomial of degree at most $-2 + \deg f$.  The integral of $\frac {f' d x}{f}$ is an integer multiple of $2 \pi i$.

For our purposes, the most interesting part of the integral comes from
\begin{equation}
  \int_\gamma \frac {\tilde{r}(x) d x}{f(x)}
\end{equation}
for $\deg \tilde{r} \leq -2 + \deg f$.  For now, we assume that $f$ has no repeated roots.  In that case, denoting by $\gamma_i$ a small contour that encircles one of $f$'s roots $x_i$ counterclockwise, we have
\begin{equation}
  \int_{\gamma_i} \frac {\tilde{r}(x) d x}{f(x)} = 2 \pi i \frac {\tilde{r}(x_i)}{f'(x_i)}.
\end{equation}
Therefore, the possible algebraic prefactors are (after dropping factors of $2 \pi i$)
\begin{equation}
  \frac {x_i^p}{f'(x_i)},
\end{equation}
for $p = 0, \dotsc, -2 + \deg f$, for $i = 1, \dotsc, \deg f$.  However, not all these quantities are independent.  Indeed, by the residue theorem we have
\begin{equation}
  \sum_{i = 1}^n \frac {x_i^p}{f'(x_i)} = 0,
\end{equation}
for $p = 0, \dotsc, -2 + \deg f$.  In total we have $(-1 + \deg f)^2$ independent quantities.  This can be seen as a pairing of $-1 + \deg f$ homology classes with $-1 + \deg f$ cohomology classes.

Notice that the Landau analysis applies to this case in the sense that the Landau equations instruct us to solve $f(x) = 0$ and $f'(x) = 0$.  The solutions to the first equations are $x = x_i$.  Then the singularities occur at $f'(x_i) = 0$, exactly as derived above.

We can also analyze cases in which $f$ has repeated roots.  If $f$ has a repeated root, it has what is called a \emph{permanent pinch}. While it is possible that the integration contour will avoid the repeated root, there is the possibility of it being pinched for all values of the polynomial coefficients for which the double root exists.  Then, the integral can become divergent (although in higher dimensions the contour may be pinched while the integral is still convergent).

Consider for simplicity the case of a polynomial $f$ with a single root of multiplicity $k$ at $x = x_0$.  Then, we have
\begin{gather}
  f(x) = \frac 1 {k!} (x - x_0)^k f^{(k)}(x_0) + \cdots, \\
  \begin{split}\tilde{r}(x) = \tilde{r}(x_0) + (x - x_0) \tilde{r}'(x_0) + \cdots +\\ + \frac 1 {(k - 1)!} (x - x_0)^{k - 1} \tilde{r}^{(k - 1)}(x_0) + \cdots,\end{split}
\end{gather}
with $f^{(k)}(x_0) \neq 0$. We then have that
\begin{gather}
  \int_{\gamma_0} \frac {\tilde{r}(x) d x}{f(x)} = 2 \pi i k \frac {\tilde{r}^{(k - 1)}(x_0)}{f^{(k)}(x_0)}.
\end{gather}
We highlight that this case is \emph{not} covered by the usual Landau equations.

\section{Imposing Integrability}
\label{app:integrability}
To every iterated integral over $\d \log$ forms, we can associate a symbol
\begin{equation}
   \sum c_{i_1,i_2,\ldots,i_n} \, L_{i_1} \otimes
\cdots \otimes L_{i_n}
\end{equation}
that faithfully preserves the analytic structure of the original function, up to contributions proportional to transcendental constants. However, not every linear combination of symbol terms corresponds to a valid iterated integral. To see this, consider upgrading a candidate symbol that depends on two external variables $u$ and $v$ to an iterated integral, by integrating along some contour $\Gamma$ from an arbitrary base point $(u_0,v_0)$ to the current values $(u,v)$:
\begin{equation}
  f (u, v ; \Gamma) = c_{i_1,i_2,\ldots,i_n} \int_{\Gamma} d \ln L_{i_1} \circ \cdots
  \circ d \ln L_{i_n} \, .
  \end{equation}
For this to integrate into a well-defined function, the iterated integral should be independent of local deformations of the path. Path independence implies that $[\partial_u, \partial_v] f = 0$. As
derivatives only act on the last entry of the symbol, we have that
\begin{equation}
  \partial_u [L_{i_1} \otimes \cdots \otimes L_{i_n}] = (\partial_u \ln L_{i_n})
  [L_{i_1} \otimes \cdots \otimes L_{i_{n - 1}}].
\end{equation}
The next derivative then acts on either $\partial_u \ln L_{i_n}$ or $L_{i_{n - 1}}$.
Therefore,
\begin{equation}
  [\partial_u, \partial_v] S \otimes L_i \otimes L_j = J_{i j} S
\end{equation}
where $J_{i j} = \partial_u \ln L_i \partial_v \ln L_j - \partial_u \ln L_j
\partial_v \ln L_i$ is the Jacobian (that is, $d \ln L_i \wedge d \ln L_j =_{}
J_{i j} d u \wedge d v$). Thus we need to find linear combinations for which
$c_{i j} J_{i j} = 0$. Demanding commutation of additional derivatives in $u$
and $v$ implies that every pair of successive symbol letters should be
integrable.

There are many ways to solve the integrability constraints. Here is one
method, which we describe for the two-loop double box alphabet from Table~\ref{tab:double_box_letters}. We start by considering the space of possible weight-two symbols. To do so, we fist compute $J_{i j}$ for all pairs of
letters, and multiply by the least common multiple of the denominators that appear. This
makes all the $J_{i j}$ polynomials in $u, v, \beta_u, \beta_v$, and $\beta_{u v}$, where
each of the square roots appears only linearly. For the
double box letters, this gives rise to 110 possible independent monomials. We can thus think of the derivative operator $[\partial_u, \partial_v]$ as mapping each weight-two symbol term
$L_i \otimes L_j$ to a 110-dimensional vector, whose components correspond to the coefficients of these 110 monomials.  We want to find the linear relations between these vectors.  An easy way to find these relations is to put
the vectors in a matrix and perform Gaussian elimination. Any zero row in the
resulting matrix corresponds to a relation. Using this method, we find $109$
independent weight-two symbols.

To construct the integrable spaces of symbols at weights three and four, we can apply a similar algorithm. Denoting the 109 independent weight-two symbols by $S_k$, we can construct a candidate weight-three space by attaching all possible letters $L_i$ to either the beginning or end of each symbol $S_k$. The set of such objects $L_i \otimes
S_k$ are then integrable in the second two
entries, while the set of such objects $S_k \otimes L_j$ are integrable in the first two entries.
The space of integrable weight-three symbols are then given by the linear combinations of symbol terms live inside of both of these spaces. This intersection can be found in the same way as for weight
two: stack the vectors in a large matrix and reduce. The non-zero rows are a
basis for the union of the spaces and the zero rows are the intersection.
We find $859$ integrable weight-three symbols. Repeating the
procedure at weight four, we find $6993$ integrable combinations of the 12 letters, as reported in Table~\ref{tab:bootstrap_dbox}. All the other constraints can be efficiently solved by finding intersections of vector spaces in this manner.

\bibliographystyle{apsrev4-1}
\bibliography{Refs}

\end{document}